\documentclass[11pt]{article}

\usepackage{enumerate}
\usepackage{amsmath,amssymb,amsfonts,amscd,amsthm,amsopn,epsfig}
\usepackage{graphicx}  
\usepackage{amsthm}
\usepackage{amssymb}
\usepackage{latexsym,stmaryrd}
\usepackage{float}
\usepackage{youngtab,ytableau}

\usepackage{hyperref}
\hypersetup{citecolor=blue}

\newtheorem{prop}{Proposition}[section]
\newtheorem{thrm}{Theorem}[section]
\newtheorem{lemma}{Lemma}[section]
\theoremstyle{definition}
\newtheorem{remark}{\it Remark}[section]
\newtheorem{example}{\it Example}[section]
\newtheorem{crl}{Corollary}[section]

\normalsize
\numberwithin{equation}{section}

\usepackage{tikz}
\usetikzlibrary{matrix,arrows,shapes}
\newcommand{\btp}{\begin{tikzpicture}[baseline=-5pt,scale=0.3,line width=0.6pt]}
\newcommand{\etp}{\end{tikzpicture}}

\newcommand{\btpm}{\begin{tikzpicture}[baseline=-5pt,scale=0.29,line width=0.6pt]}
\newcommand{\etpm}{\end{tikzpicture}}

\newcommand{\el}{\nonumber\\}

\newcommand{\CC}{\mathbb{C}}
\newcommand{\CCx}{\mathbb{C}^{\times}}

\newcommand{\cH}{\mathcal{H}}
\newcommand{\cT}{\mathcal{T}}
\newcommand{\cE}{\mathcal{E}}
\newcommand{\cV}{\mathcal{V}}

\newcommand{\chH}{\hat{\cH}}

\newcommand{\bY}{\bar{Y}}
\newcommand{\bB}{\bar{B}}

\newcommand{\Tr}{{\rm Tr}}
\newcommand{\si}{{\sigma_i}}
\newcommand{\sj}{{\sigma_j}}
\newcommand{\sip}{{\sigma_{i+1}}}

\newcommand{\ny}[1]{{\text{\small $#1$}}}

\newcommand{\mns}{{\text{\tiny -}}}
\newcommand{\pls}{{\text{\tiny +}}}

\newcommand{\dl}{\big\llbracket}
\newcommand{\dr}{\big\rrbracket}

\newcommand{\DL}{\Bigg\llbracket}
\newcommand{\DR}{\Bigg\rrbracket}

\newcommand{\vX}{\vec{X}}
\newcommand{\cX}{\cev{X}}

\newcommand{\cev}[1]{\reflectbox{\ensuremath{\vec{\,\,\reflectbox{\ensuremath{#1}}}}}\!\!}

\begin{document}


\thispagestyle{empty}

\begin{centering}

\begin{flushright}
DMUS-MP-13/11
\end{flushright}

\vspace{.2in}

{\Large {\bf On boundary fusion and functional relations \\ in the Baxterized affine Hecke algebra}}

\vspace{.2in}

{A.~Babichenko${}^{a,1}$ and V.~Regelskis${}^{b,2}$}\\
\vspace{.2 in}
${}^{a}${\emph{Department of Mathematics, University of York, \\
York, YO10 5DD, UK}} 
\\ \bigskip
${}^{b}${\emph{Department of Mathematics, University of Surrey, \\
Guildford, GU2 7XH, UK}}
\footnotetext[1]{{\tt babichen@weizmann.ac.il,}\quad ${}^{2}${\tt v.regelskis@surrey.ac.uk}}

\end{centering}

\begin{abstract}
We construct boundary type operators satisfying fused reflection equation for arbitrary representations of the Baxterized affine Hecke algebra. These operators are analogues of the fused reflection matrices in solvable half-line spin chain models. We show that these operators lead to a family of commuting transfer matrices of Sklyanin type. We derive fusion type functional relations for these operators for two families of representations.
\end{abstract}


\ytableausetup{boxsize=1.2em,aligntableaux=top}

\section{Introduction}

Originally introduced in the context of relativistic particles on a half line \cite{Chered84}, reflection equations have appeared in many different integrable systems ranging from spin chains \cite{Sklya88}, two dimensional quantum field theories \cite{GoZam92} and statistical lattice models (see e.g. \cite{BePearOB95,BePear00}). The more complex is the Lie algebraic realization of the symmetry of an integrable system with a boundary, the more involved the explicit forms of the K-matrices of the system. This makes it difficult to investigate the universal features of boundary integrability. There is hope that the main mathematical essence of integrability in general and boundary integrability in particular should be captured by the braid group, its generalizations and specializations \cite{IsaevRev}. A program for the formulation of integrability in terms of these more abstract algebraic objects was pushed forward in the context of both physical and mathematical problems (see \cite{ChGeX91,LevMar94,OrRam04,FuGeX94,KuMu05,Doi05,DoiMar05,Nep02,Arn04,dGN09,D09} and references therein). 
Generalizations of the braid group include Hecke, Brauer, Birman-Murakami-Wenzel (BMW) algebras, their quotients and boundary cousins. Specific realizations of these algebras obtained by baxterizing their elements give rise to integrable spin chains and lattice models. New general solutions constructed in terms of braid group specializations can be mapped to Lie algebraic representations, realized in spin chains and lattice models. In such a way, concrete implementations of new examples of boundary integrable models can be constructed.

Recently some progress was made in the realization of this program \cite{IsMoOg11,IsMoOg11B,Isaev05,Isaev10,IMO}. In particular, a fusion procedure for a simple class of representations of the affine baxterized Hecke algebra was developed in \cite{Isaev05,Isaev10}, and some elementary functional relations for the corresponding transfer matrix type objects were obtained. Moreover, a new solution of the baxterized reflection algebra was put forward in \cite{Isaev05}. It seems interesting to generalize this approach to a maximally wide class of Hecke algebra representations, with the goal of finding the universal hierarchy of functional relations. 
Such functional relations, combined with the analiticity input coming from the concrete realization of a given integrable model, were shown to be an effective tool in finding exact solutions of integrable models, both with closed \cite{KP92,KlP92,ZP95} and open boundary conditions \cite{Z95,N02}. Using these techniques one can avoid the direct diagonalization of the commuting set of transfer matrix type objects; some typical examples of such methods can be found in the review \cite{KNS11}. To find a hierarchy of such functional relations for a maximally general class of Hecke algebra representations is the objective of this paper.
In the case of the Lie algebraic invariant transfer matrix functional relations, known as Hirota relations for $sl_N$, this requires the knowledge of the transfer matrices and fusion hierarchy for representations corresponding to the rectangular Young diagrams. Motivated by a Lie algebraic analogy, we develop a fusion procedure, both for bulk and boundary baxterized Hecke elements, for {\it any} Hecke algebra representations, thus extending the procedure of \cite{Isaev10}. This gives a basis for the consideration of transfer matrix type functional relations demonstrated in this paper for two distinct classes of Young tableaux. 

The plan of the paper is as follows: In Section 2 we give the necessary preliminaries for both Hecke and affine Hecke algebras, and recall the construction of the primitive idempotent operators for arbitrary Young tableaux in terms of Jucys-Murphy elements, and using the Isaev-Molev-Os`kin method. In Section 3 we construct operators in the Hecke algebra that are solutions of the fusion-type bulk Yang-Baxter equation. These operators serve as building blocks of the algebraic structures considered in further sections. In Section 4 we construct boundary operators in the affine Hecke algebra which are solutions of the fusion-type reflection equation. Section 5 contains the main results of this paper. Here we construct a family of commutative transfer matrices and derive functional relations for these operators for two distinct families of representations, namely, a recursive relation for the totally--symmetric and antisymmetric representations, and functional relations for the hook--symmetric and antisymmetric representations. We conclude by explicitly calculating transfer matrices for selected Young tableaux of small size. Appendix A contains explicit derivations of some primitive functional relations which serve as simple checks of the relations presented in Section 5. For the readers' convenience, we have summarised our notation in Appendix B.

\section{Preliminaries}

In this section we define the main objects, define our notation and collect known facts about Hecke and affine Hecke algebras, which will be used in the main part of the paper. We refer to \cite{ChaPre,OgievLect,IsaevRev} for details.

\subsection{Baxterized Hecke algebra}

The Hecke algebra of type $A_{n-1}$, conveniently denoted as $\cH_n(q)$, where $q\in\CC$ and $n\geq1$, is an associative algebra over $\CC$ with generators $\sigma_i$ ($i=1 \ldots n-1$) and defining relations
\begin{align}
\si \si^{-1} &= \si^{-1} \si = 1 , \\
\si\sip\si &= \sip\si\sip , \\
\si\sj &= \sj\si \quad\text{if}\quad|i-j|>1 , \\
(\si-q)&(\si+q^{-1}) = 0 . \label{sigma_cycle}
\end{align}
The dimension of the algebra is dim $\cH_n(q)=n!$. In the $q\to1$ limit the defining relations above correspond to those of the symmetric group $\Sigma_n$. We will further consider only `generic' values of $q$, i.e.\ $q\in\CC^\times$ and is not a root of unity. Hence we will consider $\cH_n(q)$ to be semisimple. We also set $\cH_1(q)=\CC$. We will further use a concise notation $\cH_n=\cH_n(q)$. In all what follows below we assume that number $n$ is large enough, such that any index $i$ of elementary Hecke elements $\sigma_i$ does not drop out of the range $i\leq n$ as a result of any fusion procedure.

We will use the $j$--shifted elements of Hecke algebra, by which we mean a replacement of all the
elementary elements  $\sigma_i$ appearing in any Hecke words by the elementary elements $\sigma_{i,j}=\sigma_{i+j-1}$.

Let $w_k$ denote the unique longest element in the symmetric group, and let $w_{k,j}$ be the $j$--shifted equivalent of $w_k$. Then the corresponding element $X_{w_{k,j}}\in\cH_n$ ($1\leq j<n$, $2\leq k \leq n-j+1$) is given by
\begin{align} \label{WK}
X_{w_{k,j}} &= \sigma_j (\sigma_{j+1} \sigma_j) \cdots (\sigma_{j+k-3} \cdots \sigma_{j+1}\sigma_j)(\sigma_{j+k-2}\sigma_{j+k-2} \cdots \sigma_{j+1}\sigma_j) \el
 & = (\sigma_j \sigma_{j+1} \cdots \sigma_{j+k-3}\sigma_{j+k-2}) (\sigma_j \sigma_{j+1} \cdots \sigma_{j+k-3}) \cdots  (\sigma_j \sigma_{j+1}) \sigma_j .
\end{align}
These elements satisfy the following relation
\begin{equation} \label{WK:com}
X_{w_{k,j}} \sigma_l = \sigma_{2(j-1)+k-l}\, X_{w_{k,j}} , \qquad\text{for}\qquad j\leq l<j+k-1 ,
\end{equation}
which is easily deduced by induction from \eqref{WK}. For further convenience we introduce the following elements
\begin{equation}  \label{Xvec}
\vX_{m,j}= \sigma_j\sigma_{j+1}\cdots\sigma_{j+m-1}, \qquad \cX_{m,j}= \sigma_{j+m-1}\sigma_{j+m-2}\cdots\sigma_{j} ,
\end{equation}
in terms of which \eqref{WK} becomes
\begin{equation} \label{WK2}
X_{w_{k,j}} = \cX_{1,j}\, \cX_{2,j} \cdots \cX_{k-1,j} = \vX_{k-1,j} \, \vX_{k-2,j} \cdots  \vX_{1,j} .
\end{equation}
Note that $X_{w_{2,j}} = \vec{X}_{1,j} = \cev{X}_{1,j} = \sigma_j$ and we set $X_{w_{1,j}}=\vX_{0,j}=\cX_{0,j}=1$.

One can define {\it Baxterized} elements depending on a {\it spectral} parameter
 $x\in\CC$,
\begin{equation} \label{Xbax}
X_i(x) = \si + \frac{\lambda\,x}{1-x} ,\qquad\text{where}\qquad \lambda = q-q^{-1},
\end{equation}
They satisfy the Yang-Baxter equation (YBE)
\begin{equation} \label{YBE}
X_i(x)\,X_{i+1}(x y)\,X_i(y) = X_{i+1}(y)\,X_i(xy)\,X_{i+1}(x) ,
\end{equation}
and the relation
\begin{equation} \label{sxy}
X_i(x)\, X_i(1/x) = \frac{(q^2-x)(q^{-2}-x)}{(1-x)^2} = X_i(1/x)\, X_i(x) .
\end{equation}
The relation \eqref{WK:com} for the Baxterized elements becomes
\begin{equation} \label{WK:com2}
X_{w_{k,j}}\, X_l(x) = X_{2(j-1)+k-l}(x) X_{w_{k,j}} , \qquad\text{for}\qquad j\leq l<j+k-1.
\end{equation}
After normalization of $X_i(x)$ by the factor
\begin{equation}  \label{X_norm}
f(x)=\frac{1-x}{q - x/q} ,
\end{equation}
\eqref{sxy} becomes the unitarity relation
\begin{equation}
f(x) f(1/x) X_i(x) X_i(1/x) = 1.
\end{equation}
Throughout this paper we will use the non-normalized elements $X_i(x)$ due to their invariance under the substitution $q\to-1/q$. This symmetry will play an important role as we will show in further sections.


The affine Hecke algebra $\chH_n$ is an extension of $\cH_n$ and is generated by the elements $\si$ ($i=1\ldots n-1$) and the affine elements $\tau_j$ ($j=1\ldots n$) satisfying
\begin{align}
\tau_{j+1} &= \sigma_j \tau_j \sigma_j , 					\label{def:t:1} \\
\tau_l \tau_j &= \tau_j \tau_l , 								\label{def:t:2} \\
\tau_j \si &= \si \tau_j \quad\text{if}\quad j\neq i,i+1 . 			\label{def:t:3}
\end{align}
The elements $\{\tau_k\}$, $(k=1\ldots n)$ form the maximally commutative subalgebra in $\chH_n$; the set of all symmetric functions in $\tau_k$ span linearly the centre of $\chH_n$. The relations \eqref{def:t:1} and \eqref{def:t:2} imply that elements $\si$ and  $\tau_i$ satisfy the relation
\begin{equation}
\si \, \tau_i \, \si \, \tau_i = \tau_i \, \si \, \tau_i \, \si \,.
\end{equation}

One can further introduce a Baxterized element
\begin{equation}
\tau_j(x) = X_{j-1}(\tfrac{x}{\xi_{j-1}})\cdots X_{2}(\tfrac{x}{\xi_{2}}) X_{1}(\tfrac{x}{\xi_{1}})\tau_1(x) X_{1}({x}{\xi_{1}}) X_{2}({\xi_{2}})\cdots X_{j-1}({x}{\xi_{j-1}}) ,
\end{equation}
which satisfies the reflection equation (RE)
\begin{equation}
X_j(x/y) \, \tau_j(x)\, X_j(x y)\, \tau_j(y) = \tau_j(y)\, X_j(xy)\, \tau_j(x) X_j(x/y)\, . \label{RE}
\end{equation}
Here $\xi_k\in\CC^\times$ are arbitrary parameters, and  $\tau_j(x)$ is assumed to be local, $[\tau_j(x),\si]=0$, {\it c.f.} \eqref{def:t:3}. We will further always set $\xi_k=1$ without the loss of generality. The general rational solution of \eqref{RE} was found in \cite{Isaev05}
\begin{equation} \label{RE_sol}
\tau_1(x) = \frac{\tau_1 - \xi x}{\tau_1 - \xi x^{-1}} \,,
\end{equation}
where $\xi \in \CC$ is an arbitrary parameter. It is easy to see that this operator satisfies the unitarity condition,  $\tau_1(x)\tau_1(1/x)=1$, and for $\xi=0$ becomes trivial,  $\tau_1(0)=1$.

To finalize this section we want to give a diagrammatic representation of the elements of the affine algebra $\chH_n$:
\begin{equation*}
\tau_1 = \btp
	\node[] at (0,-1.5) {\tiny $1$};
	\draw [black,thick] plot [smooth,tension=.7] coordinates {(0,-1) (-0.3,-.4) (-1.2,-.2) (-1.2,.2) (-0.3,0.4) (0,1)};
	\draw [gray,ultra thick] (-0.85,-1.4) -- (-0.85,-.49);
	\draw [gray,ultra thick] (-0.95,-1.4) -- (-0.95,-.49);
	\draw [gray,ultra thick] (-0.85,1.4) -- (-0.85,-.07);
	\draw [gray,ultra thick] (-0.95,1.4) -- (-0.95,-.07);
	\etp ,
\qquad
\si = \btp
	\node[] at (1,-1.5) {\tiny $i\quad i\!+\!1$};
	\draw [black,thick] plot [smooth,tension=1] coordinates {(0,-1) (.32,-.3) (1.3,0.3) (1.6,1)};
	\draw [thick] plot [smooth,tension=1] coordinates {(1.6,-1) (1.52,-0.55) (1.18,-.2)};
	\draw [thick] plot [smooth,tension=1] coordinates {(0,1) (.08,0.55) (.41,.2)};
	\etp \!,
\qquad
\si^{-1} = \btp
	\node[] at (1,-1.5) {\tiny $i\quad i\!+\!1$};
	\draw [black,thick] plot [smooth,tension=1] coordinates {(0,1) (.32,.3) (1.3,-.3) (1.6,-1)};
	\draw [thick] plot [smooth,tension=1] coordinates {(1.6,1) (1.52,0.55) (1.18,.2)};
	\draw [thick] plot [smooth,tension=1] coordinates {(0,-1) (.08,-.55) (.41,-.2)};
	\etp \!\!,
\end{equation*}
where $i=1,\,\ldots,\,n-1$ and the vertical grey line for the element  $\tau_1$ represents a boundary pole which we will assume to be closed at infinity.

\subsection{Idempotents and Young tableaux}

Here we recall construction of the primitive orthogonal idempotents for the Hecke algebra $\cH_n$. These idempotents are conveniently defined in terms of the Jucys-Murphy elements  \cite{Jucys66,Murphy81} and are well known. We refer to \cite{OgievLect} for a pedagogical introduction. We will also state an alternative construction of these idempotents in terms of the Baxterized elements \cite{IMO}.

Let $\lambda=\{\lambda_1,\lambda_2,\ldots,\lambda_l\}$ be an $l$-partition of $m$ denoted by $\lambda\vdash m$, such that $\lambda_1\geq\lambda_2\geq\ldots\geq\lambda_l$ are non-negative numbers satisfying $\sum \lambda_i = m$. For a given $l$-partition $\lambda$ one assigns an $m$-node up-down {\it Young diagram} $\Lambda$ of shape $\lambda$ consisting of $l$ left-justified  lines with $\lambda_i$ nodes at each line. A node $\alpha$ is called {\it removable} if a diagram $\Lambda$ with the node $\alpha$ removed is still a Young diagram. A node $\beta$ is called {\it addable} if a diagram $\Lambda$ with the node $\beta$ added is a Young diagram. We will call $\cE_\pm(\Lambda)$ the set of all addable/removable nodes for the diagram $\Lambda$.

The {\it quantum contents} of a Young diagram are numbers $c_{(ij)}=q^{2(j-i)}$, where $i=1 \ldots \lambda_i,\;j=1 \ldots l$ are numbers of rows and columns of the $\alpha_{(ij)}$ node, respectively. The corresponding hook number is $h_{\alpha}=\lambda_i + \lambda'_j-i-j+1$, where the short-hand notation is $\alpha=\alpha_{(ij)}$, and $\lambda'_j$ is the number of nodes in the $j$-th column of $\lambda$. One can further fill the nodes with numbers $1 \ldots m$ called {\it contents} of nodes, turning a given Young diagram into {\it Young $m$-tableau}. A Young $m$-tableau is called {\it standard} if numbers in the nodes are increasing in left-right and up-down directions. We will denote a standard Young $m$-tableau of shape $\lambda$ by $\cT_m$. An $m$-tableau $\cT_m$ has $m-1$ subtableaux $\cT_{i}$ that are obtained by removing the last node,
\begin{equation}
\cT_{1} \subset \cT_{2} \subset \ldots \subset \cT_{m-1} \subset \cT_m \,.
\end{equation}
In the cases where needed, we will assume $\cT_{0}=1$. The quantum contents of an $m$-tableau will simply be denoted by $c_{i}$ with $i=1 \ldots m$. The sum of all standard inequivalent tableaux $\cT_m$ of shape $\lambda$ give a Young diagram $\Lambda$. Each $\cT_m$ is called a basis vector of $\Lambda$. 
\begin{example}

Choose $m=3$. Then the set of all possible partitions of $m$ are $\{\lambda\}=\{3\},\{2,1\},\{1,1,1\}$. They correspond to the following Young diagrams (with their quantum contents $c_{i}$ inside)
\begin{equation*}
\ytableaushort{{\text{\small 1}}{{\text{\scriptsize $q^2$}}}{{\text{\scriptsize $q^4$}}}}\;,\qquad
\ytableaushort{{\text{\small 1}}{{\text{\scriptsize $q^2$}}},{{\text{\scriptsize $q^{\text{-}2}$}}}}\;,\qquad
\ytableaushort{{\text{\small 1}},{{\text{\scriptsize $q^{\text{-}2}$}}},{{\text{\scriptsize $q^{\text{-}4}$}}}}\;.
\end{equation*}
The removable ($-$) and addable ($+$) nodes for each diagram are
\begin{equation*}
\ytableaushort{\hfil\hfil-+}\;,\qquad
\ytableaushort{\hfil-+,-+,+}\;,\qquad
\ytableaushort{\hfil+,\hfil,-,+}\;.
\end{equation*}
The set of all standard Young tableaux are
\begin{equation*}
\ytableaushort{123}\;,\qquad
\ytableaushort{12,3}\;,\qquad
\ytableaushort{13,2}\;,\qquad
\ytableaushort{1,2,3}\;.
\end{equation*}
\end{example}


The $j$--shifted {\it Jucys-Murphy} elements $J_{i,j}\in\cH_n$ ($1\leq i \leq n-j+1$) are defined inductively by
\begin{equation}
J_{1,j} = 1, \qquad J_{i+1,j} = \sigma_{j+i-1} \, J_{i,j} \, \sigma_{j+i-1}, 
\end{equation}
and satisfying
\begin{equation}
\qquad J_{i,j} \, \sigma_k = \sigma_k \, J_{i,j} \quad\text{if}\quad k\neq j-1,\,j+i-2,\,j+i-1.
\end{equation}
The Jucys-Murphy elements are mutually commutative and generate the maximal commutative subalgebra in $\cH_n$. 

Let $A_j^{\cT_{m}}$ be an idempotent associated with a Young tableau $\cT_m$. The idempotents $A_j^{\cT_{m}}$ can be defined inductively in the following manner
\begin{equation} \label{AT}
A_j^{\cT_{m}} = A_j^{\cT_{m-1}} \prod_{\substack{\alpha_{k}\in\cE_+({\cT_{m-1}}) \\ \alpha_{k}\neq\alpha_{m}}} \frac{J_{m,j}-c_{k}}{c_{m}-c_{k}} ,
\end{equation}
with the initial condition $A_j^{\cT_1}=1$. Furthermore,
\begin{equation}
 J_{m,j}\, A_j^{\cT_{m}} = A_j^{\cT_{m}}\, J_{m,j}  =  c_m\, A_j^{\cT_{m}}.
\end{equation}
The set of idempotents associated with all pairwise different $m$-tableaux is a partition of one,
\begin{equation}
\sum_{\cT_m} A_j^{\cT_{m}} = 1 .
\end{equation}
The sum of idempotents associated with all pairwise different tableaux $\cT_{m}$ that can be obtained from the tableau $\cT_{m-1}$ is equal to the idempotent associated with tableau $\cT_{m-1}$,
\begin{equation}
\sum_{\cT_{m}|\cT_{m-1}} A_j^{\cT_{m}} = A_j^{\cT_{m-1}} .
\end{equation}
This relation allows us derive to an alternative form of \eqref{AT}. The rational function
\begin{equation}
A_j^{\cT_{m}}(u) = A_j^{\cT_{m-1}} \frac{u-c_m}{u-J_{m,j}} ,
\end{equation}
is regular at the point $u=c_m$ and the corresponding value gives
\begin{equation}
A_j^{\cT_{m}} = A_j^{\cT_{m-1}}\, \frac{u-c_m}{u-J_{m,j}} \Big|_{u=c_m}.
\end{equation}

Let $\cT^T_m$ be a tableau obtained from $\cT_m$ by the transposition with respect to the main diagonal. Then 
\begin{equation} \label{A_trans_id}
A_j^{\cT^T_m} = A_j^{\cT_{m}} |_{q\to-q^{-1}}. 
\end{equation}

\begin{example} 
Let $\cT_m=(a_1,a_2,\ldots,a_{\lambda_1};b_1,\ldots,b_{\lambda_2};\ldots)$ denote the contents of a given $m$-tableau. Choose $m=3$. Then the set of all inequivalent idempotents $A_j^{\cT_{m}}$ is  
\begin{align} \label{A3}
A_j^{(1,2,3)} &= \frac{J_{2,j}-q^{-2}}{q^{2}-q^{-2}}\frac{J_{3,j}-q^{-2}}{q^{4}-q^{-2}}, &
A_j^{(1,2;3)} &= \frac{J_{2,j}-q^{-2}}{q^{2}-q^{-2}}\frac{J_{3,j}-q^{4}}{q^{-2}-q^{4}}, \el
A_j^{(1;2;3)} &= \frac{J_{2,j}-q^{2}}{q^{-2}-q^{2}}\frac{J_{3,j}-q^{2}}{q^{-4}-q^{2}}, &
A_j^{(1,3;2)} &= \frac{J_{2,j}-q^{2}}{q^{-2}-q^{2}}\frac{J_{3,j}-q^{-4}}{q^{2}-q^{-4}}.
\end{align}
It is a straightforward calculation to check that $A_j^{(1,2,3)}+A_j^{(1,2;3)}+A_j^{(1,3;2)}+A_j^{(1;2;3)}=1$, and $A_j^{(1;2;3)} = A_j^{(1,2,3)}|_{q\to-q^{-1}}$, $A_j^{(1,3;2)} = A_j^{(1,2;3)}|_{q\to-q^{-1}}$.

\end{example}

\begin{example} 
There are three $4$-tableaux $\cT_4=\{(1,2,4;3),(1,2;3,4),(1,2;3;4)\}$ that can be obtained from the $3$-tableau $\cT_3=(1,2;3)$. The corresponding idempotents are:
\begin{align} \label{A4}
A_j^{(1,2,4;3)} &= A_j^{\cT_3}\,\frac{J_{4,j}-1}{q^4-1}\frac{J_{4,j}-q^{-4}}{q^4-q^{-4}}, \quad
A_j^{(1,2;3,4)} = A_j^{\cT_3}\,\frac{J_{4,j}-q^4}{1-q^4}\frac{J_{4,j}-q^{-4}}{1-q^{-4}}, \el
A_j^{(1,2;3;4)} &= A_j^{\cT_3}\,\frac{J_{4,j}-q^4}{q^{-4}-q^4}\frac{J_{4,j}-1}{q^{-4}-1}.
\end{align}
It is a straightforward calculation to check that $A_j^{(1,2,4;3)}+A_j^{(1,2;3,4)}+A_j^{(1,2;3;4)}=A_j^{(1,2;3)}$.

\end{example}


Primitive idempotents for the Hecke algebra can be defined in an alternative way in terms of the Baxterized elements $X_i(x)$. This method was presented in \cite{IMO}. Here we will present a generalization which accommodates the $j$--shift. By setting $j=1$ one recovers the constructions presented in \cite{IMO}. 

Let us introduce the following operators in $\cH_n$,
\begin{align}
\vX_{m,j}(x \vec{u}) &= X_j(x u_{1})\, X_{j+1}(x u_{2}) \cdots X_{j+m-1}(x u_m) , \label{Xvec2} \\
\cX_{m,j}(x \vec{u}) &= X_{j+m-1}(x u_{1})\, X_{j+m-2}(x u_{2}) \cdots X_{j}(x u_m) ,  \label{Xvec3}
\end{align}
where we have assumed $j+m<n$; parameters $x$ and $\vec{u}=(u_1,\ldots,u_m)$  are arbitrary non-zero complex parameters. In what follows we will also be using the reversed vector notation, i.e.\ $\cev{u} = (u_m,u_{m-1},\ldots,u_1)$  and a combined notation, i.e.\ $\vec{c}\,\vec{u}=(c_1 u_1,\ldots,c_m u_m)$, $\vec{c}/\vec{u}=(c_1/u_1,\ldots,c_m/u_m)$. Let us also introduce operators
\begin{align} \label{PsiPhi}
\Psi_{m,j}(\vec{u}) = \prod_{i=1 \ldots m-1}^{\rightarrow} \! \cX_{i,j}(\vec{u}_{(i)}/u_{i+1}) \,,\qquad\quad 
\Phi_{m,j}(\vec{u}) = \prod_{i=1 \ldots m-1}^{\leftarrow} \! \vX_{i,j}(\cev{u}_{(i)}/u_{i+1}) \,,
\end{align}
where $\vec{u}_{(i)} = (u_1,\ldots,u_i)$ is a subset of the first $i$ components of the $m$-vector $\vec{u}$. We will further often use a concise notation $\vec{u}_\circ =\vec{u}_{(m-1)}$. Next, consider a rational function
\begin{equation}
f(\lambda) = \prod_{\alpha\in\lambda} \frac{\sqrt{c_\alpha}}{[h_\alpha]} , \qquad\text{where}\qquad [n] = \frac{q^n-q^{-n}}{q-q^{-1}} ,
\end{equation}
Let $\vec{c}\in\cT_m$ be the set of the quantum contents and $\lambda$ be the shape of $\cT_m$. Then the idempotent $A^{\cT_{m}}_j$ can be obtained by the consecutive evaluations
\begin{equation} \label{A_IMO}
A^{\cT_{m}}_j = f(\lambda) \, \Psi_{m,j}(\vec{u}) X^{-1}_{w_{m,j}} \big|_{u_1 = c_1} \big|_{u_2 = c_2} \ldots \big|_{u_m = c_m} .
\end{equation}
The consecutive evaluations ensure that the whole expression above is regular. We will also be in the need of an alternative definition of idempotents.

\begin{prop}
The idempotent $A^{\cT_{m}}_j$ can be obtained by the consecutive evaluations
\begin{equation} \label{A_IMO2}
A^{\cT_{m}}_j = f(\lambda) \, X^{-1}_{w_{m,j}} \Phi_{m,j}(\vec{u}) \big|_{u_1 = c_1} \big|_{u_2 = c_2} \ldots \big|_{u_m = c_m} .
\end{equation}
\end{prop}
\begin{proof}
First, note that
\begin{align}
\Psi_{m,j}(\vec{u}) &= X_j(u_1/u_2)X_{j+1}(u_1/u_3)X_j(u_2/u_3)\cdots \el
& \qquad \times X_{j+m-2}(u_{1}/u_m)\cdots X_{j+1}(u_{m-2}/u_m)X_j(u_{m-1}/u_m) \el
&= X_{j+m-2}(u_{m-1}/u_m)X_{j+m-3}(u_{m-2}/u_m)\cdots X_{j}(u_{1}/u_m) \el
& \qquad \times X_{j+m-2}(u_2/u_3)X_{j+m-3}(u_1/u_3)X_{j+m-2}(u_1/u_2) ,
\end{align}
which follows by sequences of elementary YBEs. Then, by \eqref{WK:com2},
\begin{align}
X_{j+m-2}(u_{m-1}/u_m)X_{j+m-3}(u_{m-2}/u_m)\cdots X_{j}(u_{1}/u_m) \qquad\quad\; \el
\times X_{j+m-2}(u_2/u_3)X_{j+m-3}(u_1/u_3)X_{j+m-2}(u_1/u_2) X^{-1}_{w_{m,j}} &= \el
X^{-1}_{w_{m,j}} X_{j}(u_{m-1}/u_m)X_{j+1}(u_{m-2}/u_m)\cdots X_{j+m-2}(u_{1}/u_m) \qquad \el
\times X_{j}(u_2/u_3)X_{j+1}(u_1/u_3)X_{j}(u_1/u_2) &= X^{-1}_{w_{m,j}} \Phi_{m,j}(\vec{u}) .
\end{align}
This gives $\Psi_{m,j}(\vec{u}) X^{-1}_{w_{m,j}} = X^{-1}_{w_{m,j}} \Phi_{m,j}(\vec{u})$ and by \eqref{A_IMO} we obtain \eqref{A_IMO2} as required.
\end{proof}
The primitive idempotents can also be obtained recursively. Let $\mu$ be the shape of $\cT_{m-1}$. Define a rational function
\begin{equation}
F_m(c_m) = f(\lambda) f^{-1}(\mu) .
\end{equation}
Then
\begin{align} \label{A_IMO_rec}
A^{\cT_{m}}_j &= F_m(c_m)\, A^{\cT_{m-1}}_j\,  X^{}_{w_{m-1,j}}\vX_{m-1,j}(\vec{c}_\circ/u)\, X^{-1}_{w_{m,j}} \big|_{u=c_m} \el
&= F_m(c_m) \, X^{-1}_{w_{m,j}} \, \vX_{m-1,j}(\cev{c}_{m-1}/u)\, A^{\cT_{m-1}}_j  \big|_{u=c_m} .
\end{align}

We will further use a concise notation $\big|_{\vec{u}=\vec{c}}\,$ for the consecutive evaluation. We will also use symbol $\cV_m$ (resp.\ $\cV'_n$) to denote an operator containing the set $\vec{u}$ (resp.\ $\vec{v}$) as its arguments.

\begin{prop}
Let $\vec{c}$ be the quantum contents of a tableau $\cT_m$. Then operators $\Psi_{m,j}(\vec{u})$ and $\Phi_{m,j}(\vec{u})$ satisfy the following identities,
\begin{equation} \label{AAA}
 \Phi_{m,j}(\vec{u}) \, A^{\cT_m}_{j} \big|_{\vec{u}=\vec{c}}  = \Phi_{m,j}(\vec{u}) \big|_{\vec{u}=\vec{c}} \,, \qquad\quad 
 A^{\cT_m}_{j} \, \Psi_{m,j}(\vec{u}) \big|_{\vec{u}=\vec{c}} = \Psi_{m,j}(\vec{u}) \big|_{\vec{u}=\vec{c}}  \,.
\end{equation}
\end{prop}
\begin{proof}
The idempotence property $(A_{j}^{\cT_{m}})^2 = A_{j}^{\cT_{m}}$ together with the definitions \eqref{A_IMO} and \eqref{A_IMO2} give
\begin{align}
f(\lambda)\, \Phi_{m,j}(\vec{u}) \, X^{-1}_{w_{m,j}} \Phi_{m,j}(\vec{u}) \big|_{\vec{u}=\vec{c}} &= \Phi_{m,j}(\vec{u}) \big|_{\vec{u}=\vec{c}} \;, \\
f(\lambda)\, \Psi_{m,j}(\vec{u}) \, X^{-1}_{w_{m,j}} \Psi_{m,j}(\vec{u}) \big|_{\vec{u}=\vec{c}} &= \Psi_{m,j}(\vec{u}) \big|_{\vec{u}=\vec{c}} \;.
\end{align}
which are equivalent to \eqref{AAA}.
\end{proof}

Primitive idempotents for the totally symmetric ($\vec{c}=(1,q^2,q^4,\ldots,q^{2m-2})$) and totally antisymmetric ($\vec{c}=(1,q^{-2},q^{-4},\ldots,q^{2-2m})$) tableaux have the following form
\begin{equation} \label{A_IMO_sa}
A^{\cT_m}_j = \frac{\Phi_{m,j}(\vec{c})}{[m]_{\pm}!} = \frac{\Psi_{m,j}(\vec{c})}{[m]_{\pm}!} ,
\end{equation}
where $[m]_{\pm}!=[m]_{\pm}\cdots[2]_{\pm}[1]_{\pm}$ is the $q$-factorial, and $[m]_{\pm} = [m]\,|_{q\to\pm q}$. Here the plus sign stands for the symmetric case and minus for the antisymmetric one. These expressions contain no zeros in the individual elements and can be evaluated all at once.

\smallskip

\begin{example}
The hook numbers $h_\alpha$ for all $4$-node Young diagrams are 
\begin{equation*}
\ytableaushort{4321}\;,\qquad
\ytableaushort{421,1}\;,\qquad
\ytableaushort{32,21}\;,\qquad
\ytableaushort{41,2,1}\;,\qquad
\ytableaushort{4,3,2,1}\;.
\end{equation*}
\end{example}

\begin{example}
Idempotents for $\cT_4$ are found by evaluating
\begin{align} \label{A4_IMO}
 A^{\cT_{4}}_j &= f(\lambda)\, X_j(u_1/u_2) X_{j+1}(u_1/u_3) X_j(u_2/u_3) X_{j+2}(u_1/u_4) X_{j+1}(u_2/u_4) X_j(u_3/u_4) \, X^{-1}_{w_{4,j}} \el
 &= f(\lambda)\, X^{-1}_{w_{4,j}} X_j(u_3/u_4) X_{j+1}(u_2/u_4) X_{j+2}(u_1/u_4) X_j(u_2/u_3) X_{j+1}(u_1/u_3) X_j(u_1/u_2) ,
\end{align}
at $\vec{u}=\vec{c}\in\cT_4$. Here $X_{w_{4,j}} = \sigma_j\sigma_{j+1}\sigma_j\sigma_{j+2}\sigma_{j+1}\sigma_j$.  In particular, for $\cT_4=(1,2;3,4)$ we have $\vec{c}=(1,q^2,q^{-2},1)$ and $f(\lambda) = 1/\big([3]\,[2]^2 \big)$,  for $\cT_4=(1,3;2,4)$ we have $\vec{c}=(1,q^{-2},q^{2},1)$ and $f(\lambda)$ is the same. 
\end{example}

\section{Fused solutions of bulk Yang-Baxter equation} \label{sec:3}

In this section, we introduce simple concepts which are crucial, as building blocks of the algebraic structures, to our subsequent analysis of the affine Hecke algebra. Here we construct fused operators with values in the Hecke algebra which are solutions of the fused Yang-Baxter equations (simply called {\it fused YBEs} for convenience). We then describe some properties of these fused operators. (The authors believe that operators of this type must be well-known, but were unable to locate them explicitly in the literature. The closest construction we found was in \cite{BePear00} where such operators, the Fused Row Operators and corresponding fused YBEs, for the totally symmetric representations were considered.) 
\begin{prop} \label{prop31}
In the Hecke algebra the following identities, fused Yang-Baxter equations, hold $(\forall x,z;\; \forall \vec{u})$
\begin{align}
X_{j}(x/z)\, \vX_{m,j+1}(x \vec{u})\, \vX_{m,j}(z \vec{u}) &= \vX_{m,j+1}(z \vec{u})\, \vX_{m,j}(x \vec{u})\, X_{j+m}(x/z) \,, \label{fYBE1} \\
X_{j+m}(x/z)\, \cX_{m,j}(x/\cev{u})\, \cX_{m,j+1}(z/\cev{u}) &= \cX_{m,j}(z/\cev{u})\, \cX_{m,j+1}(x/\cev{u})\, X_{j}(x/z) \,, \label{fYBE2} \\
\vX_{m,j}(x \vec{u})\, X_{j+m}(xz)\, \cX_{m,j}(z/\cev{u}) &= \cX_{m,j+1}(z/\cev{u})\, X_{j}(xz)\, \vX_{m,j}(x \vec{u}) \,. \label{fYBE3}
\end{align}
\end{prop}
\begin{proof}
We will prove the fused YBEs above by induction. The base of induction is \eqref{YBE}. Suppose $\vX_{m-1,j}(x \vec{u}_\circ)$ satisfies \eqref{fYBE1}. Then the RHS of \eqref{fYBE1} can be written as
\begin{align}
&\vX_{m,j+1}(z \vec{u})\, \vX_{m,j}(x \vec{u})\, X_{j+m}(x/z) \el
& \qquad = \vX_{m-1,j+1}(z \vec{u}_\circ)\, \vX_{m-1,j}(x \vec{u}_\circ)\,X_{j+m}(z u_m)\,X_{j+m-1}(x u_m)\, X_{j+m}(x/z) \el
& \qquad = X_{j}(x/z)\,\vX_{m-1,j+1}(x \vec{u}_\circ)\, \vX_{m-1,j}(z \vec{u}_\circ)\, X_{j+m}(x u_m)\,X_{j+m-1}(z u_m)  \el
& \qquad = X_{j}(x/z)\,\vX_{m,j+1}(x \vec{u})\, \vX_{m,j}(z \vec{u}) \,,
\end{align}
where in the second equality we have used \eqref{YBE} together with the induction hypothesis. The proof of \eqref{fYBE2} is analogous,
\begin{align}
& \cX_{m,j}(z/\cev{u})\, \cX_{m,j+1}(x/\cev{u})\, X_{j}(x/z)\el
& \qquad = X_{j+m-1}(z/u_m)\,X_{j+m}(x/u_m)\,\cX_{m-1,j}(z/\cev{u}_\circ)\, \cX_{m-1,j+1}(x/\cev{u}_\circ)\, X_{j}(x/z) \el
& \qquad = X_{j+m}(x/z)\, X_{j+m-1}(x/u_m)\,X_{j+m}(z/u_m)\,\cX_{m-1,j}(x/\cev{u}_\circ)\,\cX_{m-1,j+1}(z/\cev{u}_\circ) \el
& \qquad = X_{j+m}(x/z)\,\cX_{m,j}(x \vec{u})\,\cX_{m,j+1}(z \vec{u})  \,.
\end{align}
For \eqref{fYBE3} we have
\begin{align}
& \vX_{m,j}(x \vec{u})\, X_{j+m}(xz)\, \cX_{m,j}(z/\cev{u}) \el
& \qquad = \vX_{m-1,j}(x \vec{u}_\circ)\,X_{j+m-1}(x u_m)\, X_{j+m}(xz)\,X_{j+m-1}(z/u_m)\, \cX_{m-1,j}(z/\cev{u}_\circ) \el
& \qquad = X_{j+m}(z/u_m)\, \cX_{m-1,j+1}(z/\cev{u}_\circ)\,  X_{j}(xz)\,\vX_{m-1,j+1}(x \vec{u}_\circ)\,  X_{j+m}(x u_m) \el
& \qquad = \cX_{m,j+1}(z/\cev{u})\, X_{j}(xz)\, \vX_{m,j+1}(x \vec{u}) \,.
\end{align}
\end{proof}
Let us introduce conjugation $C$, inverse $I$ and reverse $R$ operators which act on the set $\vec{u}$ as follows
\begin{align}
 C:\quad & \vec{u}\to1/\cev{u}, \quad u_i\mapsto1/u_{m-i+1} \,, \label{map:C} \\
 I:\quad & \vec{u}\to1/\vec{u}, \quad u_i\mapsto1/u_{i} \,, \label{map:I} \\
 R:\quad & \vec{u}\to \cev{u}, \qquad u_i\mapsto u_{m-i+1} \,. \label{map:R}
\end{align}
It is easy to see that these maps are symmetries of the fused YBEs defined above. These operators will be employed in the next section.
We will further give a specialization of the fused YBEs for a given tableau $\cT_m$ (and $\vec{c}\in\cT_m$).
\begin{prop} \label{prop32}
Fused projected operators $(\forall x, \forall \cT_m)$
\begin{equation} \label{Xpm:YBE}
X_{j,+}^{\cT_{m}}(x) = A_{j+1}^{\cT_{m}}\, \vX_{m,j}(x\,\vec{c}\,) \,,  \quad\qquad
X_{j,-}^{\cT_{m}}(x) = A_{j}^{\cT_{m}}\, \cX_{m,j}(x/\cev{c}\,) \,,
\end{equation}
are solutions of the fused projected Yang-Baxter equations
\begin{align}
X_{j}(x/z)\, X_{j+1,+}^{\cT_{m}}(x)\, X_{j,+}^{\cT_{m}}(z) &= X_{j+1,+}^{\cT_{m}}(z)\, X_{j,+}^{\cT_{m}}(x)\, X_{j+m}(x/z) \,, \label{fYBE1A} \\
X_{j+m}(x/z)\, X_{j,-}^{\cT_{m}}(x)\, X_{j+1,-}^{\cT_{m}}(z) &= X_{j,-}^{\cT_{m}}(z)\, X_{j+1,-}^{\cT_{m}}(x)\, X_{j}(x/z) \,, \label{fYBE2A} \\
X_{j,+}^{\cT_{m}}(x)\, X_{j+m}(xz)\, X_{j,-}^{\cT_{m}}(z) &= X_{j+1,-}^{\cT_{m}}(z)\, X_{j}(xz)\, X_{j+1,+}^{\cT_{m}}(x) \,. \label{fYBE3A}
\end{align}
\end{prop}
\noindent To prove the proposition above we will need the following Lemma.
\begin{lemma} \label{lemma31}
In the Hecke algebra the following identities hold $(\forall x, \forall \cT_m)$,
\begin{align}
A_{j+1}^{\cT_{m}}\, \vX_{m,j}(x\,\vec{c}\,)\, A_{j}^{\cT_{m}} &= A_{j+1}^{\cT_{m}}\, \vX_{m,j}(x\,\vec{c}\,) \,, \label{lemma31:1} \\
A_{j}^{\cT_{m}}\, \cX_{m,j}(x/\cev{c}\,)\, A_{j+1}^{\cT_{m}} &= A_{j}^{\cT_{m}}\, \cX_{m,j}(x/\cev{c}\,) \,. \label{lemma31:2} 
\end{align}
\end{lemma}

\begin{proof}
First, observe that
\begin{align}
\Phi_{m,j+1}(\vec{u}) \, \vX_{m,j}(x\, \vec{u}) &= \vX_{m,j}(x\, \cev{u}) \, \Phi_{m,j}(\vec{u}) \,,  \label{FXXF1} \\
\Phi_{m,j}(\vec{u}) \, \cX_{m,j}(x/ \cev{u}) &= \cX_{m,j}(x/ \vec{u}) \, \Phi_{m,j+1}(\vec{u}) \,, \label{FXXF2}
\end{align}
which follow by sequences of YBEs \eqref{YBE}. Then the LHS of \eqref{lemma31:1} can be written as
\begin{align}
A_{j+1}^{\cT_{m}}\, \vX_{m,j}(x\,\vec{c}\,)\, A_{j}^{\cT_{m}} &= f(\lambda)\, X^{-1}_{w_{m,j+1}} \vX_{m,j}(x\,\cev{u})\, \Phi_{m,j}(\vec{u}) \, A_{j}^{\cT_{m}} \big|_{\vec{u}=\vec{c}} \el
&= f(\lambda)\, X^{-1}_{w_{m,j+1}} \vX_{m,j}(x\,\cev{u})\, \Phi_{m,j}(\vec{u}) \big|_{\vec{u}=\vec{c}} \, = A_{j+1}^{\cT_{m}}\, \vX_{m,j}(x\,\vec{c}\,) \,.
\end{align}
Here we have used \eqref{AAA}. In a similar way, the LHS of \eqref{lemma31:2} becomes
\begin{align}
A_{j}^{\cT_{m}}\, \cX_{m,j}(x/\cev{c}\,)\, A_{j+1}^{\cT_{m}} &= f(\lambda)\, X^{-1}_{w_{m,j}} \cX_{m,j}(x/\vec{u})\, \Phi_{m,j+1}(\vec{u}) \, A_{j+1}^{\cT_{m}} \big|_{\vec{u}=\vec{c}} \el
&= f(\lambda)\, X^{-1}_{w_{m,j}} \cX_{m,j}(x/\cev{u})\, \Phi_{m,j+1}(\vec{u}) \big|_{\vec{u}=\vec{c}} \, = A_{j}^{\cT_{m}}\, \cX_{m,j}(x/\cev{c}\,) \,.
\end{align}
\end{proof}
We are now ready to give a proof of Proposition \ref{prop32}.
\begin{proof}
The RHS of \eqref{fYBE1A}, by \eqref{lemma31:1}, becomes
\begin{align}
& X_{j+1,+}^{\cT_{m}}(z) \, X_{j,+}^{\cT_{m}}(x)\, X_{j+m}(x/z) = A_{j+1}^{\cT_m}\, \vX_{m,j+1}(z \vec{c}) \, \vX_{m,j}(x\vec{c})\, X_{j+m}(x/z) \el
& \qquad = A_{j+1}^{\cT_m} \, X_{j}(x/z) \, \vX_{m,j+1}(x\vec{c}) \, \vX_{m,j}(z \vec{c}) = X_{j}(x/z) \, X_{j+1,+}^{\cT_{m}}(x) \, X_{j,+}^{\cT_{m}}(z)  \,,
\end{align}
where in the second equality we have used \eqref{fYBE1}. The proof of \eqref{fYBE2A} is analogous,
\begin{align} \label{Y_fact}
& X_{j,-}^{\cT_{m}}(z)\, X_{j+1,-}^{\cT_{m}}(x)\, X_{j}(x/z) = A_{j}^{\cT_m}\, \cX_{m,j}(z/\cev{c}) \, \cX_{m,j+1}(x/\cev{c})\, X_{j}(x/z) \el
& \qquad = A_{j}^{\cT_m} \, X_{j+m}(x/z) \, \cX_{m,j}(x/\cev{c}) \, \cX_{m,j+1}(z/\cev{c}) = X_{j+m}(x/z) \, X_{j,-}^{\cT_{m}}(x) \, X_{j+1,-}^{\cT_{m}}(z)  \,,
\end{align}
where we have used \eqref{lemma31:2} and \eqref{fYBE2}. For \eqref{fYBE3A} we have
\begin{align}
& X_{j,+}^{\cT_{m}}(x)\, X_{j+m}(xz)\, X_{j,-}^{\cT_{m}}(z) = A_{j+1}^{\cT_m} \,\vX_{m,j}(x \vec{c})\, X_{j+m}(xz)\, \cX_{m,j}(z/\cev{c}) \el
& \qquad = A_{j+1}^{\cT_m} \,\cX_{m,j+1}(z/\cev{c})\, X_{j}(xz)\,\vX_{m,j+1}(x \vec{c}) = X_{j+1,-}^{\cT_{m}}(z)\, X_{j}(xz)\, X_{j+1,+}^{\cT_{m}}(x) \,,
\end{align}
where we have used \eqref{lemma31:1}, \eqref{lemma31:2} and \eqref{fYBE3}.
\end{proof}

To conclude this section we will state some useful properties of the fused projected operators and demonstrate a few examples afterwards. Fused projected operators enjoy the following transposition
\begin{equation} \label{Y_trans_id}
X_{j,\pm}^{\cT^T_m}(x) = X_{j,\pm}^{\cT_{m}}(x)|_{q\to-q^{-1}} \,, 
\end{equation}
and the pairwise idempotence identities
\begin{equation} \label{Y_idemp_id}
f^{\cT_{m}}(x)\, X_{j,+}^{\cT_{m}}(x)\, X_{j,-}^{\cT_{m}}(1/x) = A_{j+1}^{\cT_{m}} \,, \qquad\quad f^{\cT_{m}}(x)\, X_{j,-}^{\cT_{m}}(x)\, X_{j,+}^{\cT_{m}}(1/x) = A_{j}^{\cT_{m}} ,
\end{equation}
where
\begin{equation}
f^{\cT_{m}}(x) = \prod_{i=1\ldots m} \!\Big(f(x c_i) f(1/(x c_i) \Big)
\end{equation}
is the normalization factor. They follow rather straightforwardly by \eqref{A_trans_id}, Lemma \ref{lemma31}, and properties of $X_i(x)$.
%


\begin{example}
Fused operators for the tableau $\cT_4=(1,2;3,4)$ are
\begin{align}
X_{j,+}^{\cT_{m}}(x) 	&= A_{j+1}^{\cT_{m}}\, X_j(x c_{1}) \, X_{j+1}(x c_{2}) \, X_{j+2}(x c_{3}) \, X_{j+3}(x c_{4}) \,, \\
X_{j,-}^{\cT_{m}}(x) 	&= A_{j}^{\cT_{m}}\, X_{j+3}(x/c_{4}) \, X_{j+2}(x/c_{3}) \, X_{j+1}(x/c_{2}) \, X_{j}(x/c_{1}) \,,
\end{align} 
where $c_{1}=1$, $c_{2}=q^2$, $c_{3}=q^{-2}$, $c_{4}=1$, and $A_{j}^{\cT_{m}}$ is given in \eqref{A4} and \eqref{A4_IMO}.
\end{example}
\begin{example}
The set of all pairwise transposable $4$-tableaux is given below:
\smallskip
\begin{equation*}
\ytableaushort{{\ny{1}}{\ny{3}}{\ny{4}},{\ny{2}}}\;\;\overset{T}{\Longleftrightarrow}\;\;
\ytableaushort{{\ny{1}}{\ny{2}},{\ny{3}},{\ny{4}}}\;,\qquad
\ytableaushort{{\ny{1}}{\ny{2}}{\ny{3}},{\ny{4}}}\;\;\overset{T}{\Longleftrightarrow}\;\;
\ytableaushort{{\ny{1}}{\ny{4}},{\ny{2}},{\ny{3}}}\;,\qquad
\ytableaushort{{\ny{1}}{\ny{2}}{\ny{4}},{\ny{3}}}\;\;\overset{T}{\Longleftrightarrow}\;\;
\ytableaushort{{\ny{1}}{\ny{3}},{\ny{2}},{\ny{4}}}\;,
\end{equation*}
\smallskip
\begin{equation*}
\ytableaushort{{\ny{1}}{\ny{2}}{\ny{3}}{\ny{4}}}\;\;\overset{T}{\Longleftrightarrow}\;\;
\ytableaushort{{\ny{1}},{\ny{2}},{\ny{3}},{\ny{4}}}\;,\qquad
\ytableaushort{{\ny{1}}{\ny{2}},{\ny{3}}{\ny{4}}}\;\;\overset{T}{\Longleftrightarrow}\;\;
\ytableaushort{{\ny{1}}{\ny{3}},{\ny{2}}{\ny{4}}}\,.
\end{equation*}
\end{example}

\section{Fused solutions of affine Hecke reflection equation} \label{sec:4}

In this section we construct fused boundary operators, with values in the affine Hecke algebra, satisfying the fused reflection equation ({\it fused RE}). Following the same strategy as in Section 3, we will first define a fused boundary operator of a generic type satisfying the fused RE. We will then give a specialization of this fused boundary operator for a given tableau $\cT_m$. In the next section these operators will lead to transfer matrices, which are analogues of the Sklyanin--type transfer matrices and generalizations of the elements considered in \cite{Isaev05,Isaev10}. Fused boundary operators of a similar type for arbitrary representations were also considered in \cite{IsMoOg11}.%

Let us introduce fused Yang-Baxter operators ($\vec{u}\in\cV_m,\; \vec{v}\in\cV_n'$)
\begin{align}
Y_{j,++}^{\cV_m,\cV'_n}(x) &= \prod_{i=1\ldots m}^{\leftarrow} \vX_{n,j+i-1}(xu_{m-i+1}\vec{v}\,) \;= \prod_{i=1\ldots n}^{\rightarrow} \cX_{m,j+i-1}(x\vec{u}\,v_i) \,, \label{Ypp} \\
Y_{j,+-}^{\cV_m,\cV'_n}(x) &= \prod_{i=1\ldots m}^{\leftarrow} \vX_{n,j+i-1}(xu_{m-i+1}/\cev{v}) = \prod_{i=1\ldots n}^{\rightarrow} \cX_{m,j+i-1}(x\vec{u}/v_{n-i+1}) \,, \label{Ypm} \\
Y_{j,-+}^{\cV_m,\cV'_n}(x) &= \prod_{i=1\ldots m}^{\leftarrow} \vX_{n,j+i-1}(x\vec{v}/u_i) \hspace{.87cm} = \prod_{i=1\ldots n}^{\rightarrow} \cX_{m,j+i-1}(x v_i/\cev{u}) \,, \label{Ymp} \\
Y_{j,--}^{\cV_m,\cV'_n}(x) &= \prod_{i=1\ldots m}^{\leftarrow} \vX_{n,j+i-1}(x/(u_i\cev{v}\,)) \hspace{.51cm} = \prod_{i=1\ldots n}^{\rightarrow} \cX_{m,j+i-1}(x/(\cev{u}\,v_{n-i+1})) \,. \label{Ymm}
\end{align}
These fused operators are combinations of the elements given in Proposition \ref{prop31}. Their role will be clear after we introduce a Theorem below. For now let us note that elements $Y_{j,\pm\pm}^{\cV_m,\cV'_n}(x)$ are related to each other by the conjugation map \eqref{map:C} as shown in the diagram below:
\begin{center}\btp
	\node[] at (-15,0) {$Y_{j,++}^{\cV_m,\cV'_n}(x)$};	
	\node[] at (0,3) {$Y_{j,+-}^{\cV_m,\cV'_n}(x)$};		
	\node[] at (0,-3) {$Y_{j,-+}^{\cV_m,\cV'_n}(x)$};	
	\node[] at (15,0) {$Y_{j,--}^{\cV_m,\cV'_n}(x)$};	
	\node[] at (-7.5,-3) {\small $C_{\cV_m}$};
	\node[] at (-7.5,3) {\small $C_{\cV'_n}$};
	\node[] at (7.5,-3) {\small $C_{\cV'_n}$};
	\node[] at (7.5,3) {\small $C_{\cV_m}$};	
	\draw[<->,color=black] (-11,1) -- (-4,2.7);
	\draw[<->,color=black] (-11,-1) -- (-4,-2.7);
	\draw[<->,color=black] (4,2.7) -- (11,1);
	\draw[<->,color=black] (4,-2.7) -- (11,-1) ;
\etp
\end{center}
\smallskip
\begin{thrm} \label{th51} 
Fused boundary operators $(\forall x,z;\; \forall \cV_m,\cV_n')$
\begin{equation}
B^{\cV_m}_j(x) = \prod_{i=1\ldots m}^{\rightarrow} \cX_{i-1,j} (x^2 u_{i}\, \vec{u}_{(i-1)}) \, \tau_j(x\,u_{i})
\label{FusedB} 
\end{equation}
are solutions of the fused reflection equation
\begin{equation} \label{fRE}
Y_{j,-+}^{\cV'_n,\cV_m}(x/z)\,B_j^{\cV_{m}}(x)\,Y_{j,++}^{\cV_m,\cV'_n}(xz)\,B_j^{\cV'_{n}}(z) = B_j^{\cV'_{n}}(z)\,Y_{j,++}^{\cV'_n,\cV_m}(xz)\,B_j^{\cV_{m}}(x)\,Y_{j,+-}^{\cV_m,\cV'_n}(x/z) \,. 
\end{equation}
\end{thrm}
\begin{proof} 
We will prove the fused RE by induction, and use a diagrammatic representation to illustrate the proof. There are two cases that need to be considered.

\noindent {\it Case 1}. Let $B_j^{\cV_{m-1}}(x)\subset B_j^{\cV_{m}}(x)$ be a fused boundary operator for the subset $\cV_{m-1}\subset\cV_m$. Suppose that the pair of operators $B_j^{\cV_{m-1}}(x)$ and $B_j^{\cV'_{n}}(z)$ is a solution of \eqref{fRE}. Then we need to prove that the pair $B_j^{\cV_{m}}(x)$ and $B_j^{\cV'_{n}}(z)$ is also a solution of the fused RE. The proof goes in the following steps:

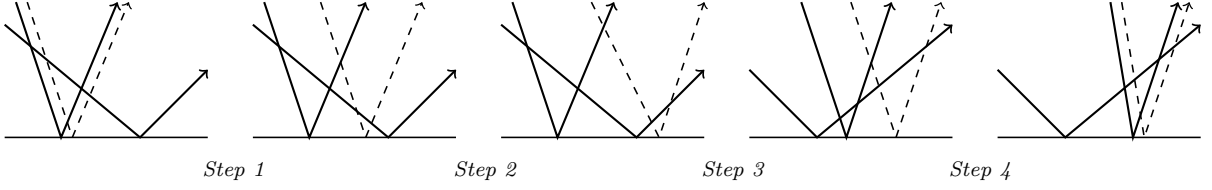
\begin{center}\btp
	\draw[-] (0,0) -- (9,0);
	\draw[->,color=black,thick] (0.5,6) -- (2.5,0) -- (5,6);
	\draw[->,color=black,thick] (0,5) -- (6,0) -- (9,3);
	\draw[->,color=black,dashed] (1,6) -- (3,0) -- (5.5,6);
	\draw[-] (11,0) -- (20,0);
	\draw[->,color=black,thick] (11.5,6) -- (13.5,0) -- (16,6);
	\draw[->,color=black,thick] (11,5) -- (17,0) -- (20,3);
	\draw[->,color=black,dashed] (14,6) -- (16,0) -- (18.5,6);
	\draw[-] (22,0) -- (31,0);
	\draw[->,color=black,thick] (22.5,6) -- (24.5,0) -- (27,6);
	\draw[->,color=black,thick] (22,5) -- (28,0) -- (31,3);
	\draw[->,color=black,dashed] (26,6) -- (29,0) -- (31,6);
	\draw[-] (33,0) -- (42,0);
	\draw[->,color=black,thick] (35.3,6) -- (37.3,0) -- (39.3,6);
	\draw[->,color=black,thick] (33,3) -- (36,0) -- (42,5);
	\draw[->,color=black,dashed] (37.5,6) -- (39.5,0) -- (41.5,6);
	\draw[-] (44,0) -- (53,0);
	\draw[->,color=black,thick] (49,6) -- (50,0) -- (52,6);
	\draw[->,color=black,thick] (44,3) -- (47,0) -- (53,5);
	\draw[->,color=black,dashed] (49.5,6) -- (50.5,0) -- (52.5,6);
	\node[right] at (8.3,-1.5) {\it \scriptsize Step 1 \hspace{2.3cm} Step 2 \hspace{2.25cm} Step 3  \hspace{2.25cm} Step 4 };
\etp
\end{center}
Here the dashed line represents the difference between the operators $B_j^{\cV_{m-1}}(x)$ and $B_j^{\cV_{m}}(x)$.

\noindent {\it Case 2}. Let $B_j^{\cV'_{n-1}}(z)\subset B_j^{\cV'_n}(z)$ be a fused boundary operator for the subset $\cV'_{n-1}\subset\cV'_n$. Suppose that the pair of operators $B_j^{\cV_{m}}(x)$ and $B_j^{\cV'_{n-1}}(z)$ is a solution of \eqref{fRE}. Then we need to prove that the pair $B_j^{\cV_{m}}(x)$ and $B_j^{\cV'_{n}}(z)$ is also a solution of the fused RE. The proof goes in the following steps:
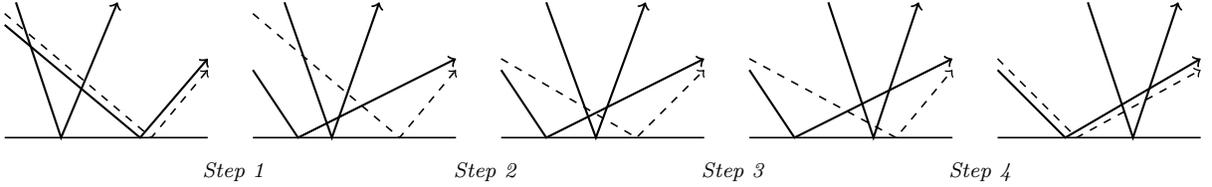
\begin{center}\btp
	\draw[-] (0,0) -- (9,0);
	\draw[->,color=black,thick] (0.5,6) -- (2.5,0) -- (5,6);
	\draw[->,color=black,thick] (0,5) -- (6,0) -- (9,3.5);
	\draw[->,color=black,dashed] (0,5.5) -- (6.5,0) -- (9,3);
	\draw[-] (11,0) -- (20,0);
	\draw[->,color=black,thick] (12.4,6) -- (14.5,0) -- (16.6,6);
	\draw[->,color=black,thick] (11,3) -- (13,0) -- (20,3.5);
	\draw[->,color=black,dashed] (11,5.5) -- (17.5,0) -- (20,3);
	\draw[-] (22,0) -- (31,0);
	\draw[->,color=black,thick] (24,6) -- (26.2,0) -- (28.4,6);
	\draw[->,color=black,thick] (22,3) -- (24,0) -- (31,3.5);
	\draw[->,color=black,dashed] (22,3.5) -- (28,0) -- (31,3);
	\draw[-] (33,0) -- (42,0);
	\draw[->,color=black,thick] (36.5,6) -- (38.5,0) -- (40.5,6);
	\draw[->,color=black,thick] (33,3) -- (35,0) -- (42,3.5);
	\draw[->,color=black,dashed] (33,3.5) -- (39.5,0) -- (42,3);
	\draw[-] (44,0) -- (53,0);
	\draw[->,color=black,thick] (48,6) -- (50,0) -- (52,6);
	\draw[->,color=black,thick] (44,3) -- (47,0) -- (53,3.5);
	\draw[->,color=black,dashed] (44,3.5) -- (47.5,0) -- (53,3);
	\node[right] at (8.3,-1.5) {\it \scriptsize Step 1 \hspace{2.3cm} Step 2 \hspace{2.25cm} Step 3  \hspace{2.25cm} Step 4 };
\etp
\end{center}
Here the dashed line represents the difference between the operators $B_j^{\cV'_{n-1}}(z)$ and $B_j^{\cV_{n}}(z)$

The base for induction is \eqref{RE}. Then by virtue of {\it Case 1} and {\it Case 2} one can obtain fused RE for any $\cV_m$ and $\cV'_n$. We will give a proof for the first case only. The proof for the second case is analogous. 

We start by factorizing fused operators in both sides of fused RE \eqref{fRE}. Fused operators in the LHS can be factorized in the following way
\begin{align} 
Y_{j,-+}^{\cV'_n,\cV_m}(x/z) &= 
Y_{j,-+}^{\cV'_n,\cV_{m-1}}(x/z) \, \cX_{n,j+m-1}(x u_{m}/(z \cev{v}\,)) \,, \label{Ymfac} \\
Y_{j,++}^{\cV_m,\cV'_n}(xz) &= 
Y_{j+1,++}^{\cV_{m-1},\cV'_n}(xz) \, \vX_{n,j}(xz u_{m}\vec{v}) \,, \label{Ypfac} \\
B_j^{\cV_{m}}(x) &=  
B_j^{\cV_{m-1}}(x) \, \cX_{m-1,j}(x^2 u_m \vec{u}_\circ) \, \tau_j(x\,u_m) \,, \label{RE:fact1}
\end{align}
giving
\begin{align} \label{RE:step0}
& Y_{j,-+}^{\cV'_n,\cV_m}(x/z) \,B_j^{\cV_{m}}(x) \, Y_{j,++}^{\cV_m,\cV'_n}(xz) \, B_j^{\cV'_{n}}(z)\el
& \qquad = Y_{j,-+}^{\cV'_n,\cV_{m-1}}(x/z) \, B_j^{\cV_{m-1}}(x) \, \cX_{n,j+m-1}(x u_{m}/(z \cev{v}\,)) \, \cX_{m-1,j}(x^2 u_m \vec{u}_\circ) \, Y_{j+1,++}^{\cV_{m-1},\cV'_n}(xz) \el
& \qquad \qquad \times \tau_j(x\,u_m) \, \vX_{n,j}(xz u_{m}\vec{v}) \, B_j^{\cV'_{n}}(z) \,.
\end{align}
Similarly, operators in the RHS can be factorized as follows,
\begin{align} \label{RE:fact2}
Y_{j,+-}^{\cV'_n,\cV_m}(xz) &= Y_{j,+-}^{\cV'_n,\cV_{m-1}} (xz)\, \cX_{n,j+m-1}(x zu_m \vec{v}) \,, \\
Y_{j,--}^{\cV_m,\cV'_n}(x/z) &= Y_{j+1,--}^{\cV_{m-1},\cV'_n}(x/z) \, \vX_{n,j}(xu_m/(z \cev{v})) \,, \\
B_j^{\cV_{m}}(x) &= B_j^{\cV_{m-1}}(x) \, \cX_{m-1,j}(x^2 u_m \vec{u}_\circ) \, \tau_j(x\,u_m) \,,
\end{align}
giving
\begin{align} \label{RE:step5}
& B_j^{\cV'_{n}}(z) \, Y_{j,+-}^{\cV'_n,\cV_m}(xz) \, B_j^{\cV_{m}}(x) \, Y_{j,--}^{\cV_m,\cV'_n}(x/z) \el
& \qquad = B_j^{\cV'_{n}}(z) \, Y_{j,+-}^{\cV'_n,\cV_{m-1}} (xz)\, \cX_{n,j+m-1}(x zu_m \vec{v}) \, B_j^{\cV_{m-1}}(x) \el
& \qquad\qquad \times \cX_{m-1,j}(x^2 u_m \vec{u}_\circ) \, \tau_j(x\,u_m) \, Y_{j+1,--}^{\cV_{m-1},\cV'_n}(x/z) \, \vX_{n,j}(xu_m/(z \cev{v})) \,.
\end{align}
\begin{figure} \begin{center}\btpm
	\draw[-] (0,0) -- (50,0);
	\draw[->] (0,30) -- (20,0) -- (40,30);
	\draw[->] (4,30) -- (24,0) -- (44,30);
	\draw[->] (8,30) -- (28,0) -- (48,30);
	\draw[->,dashed] (12,30) -- (32,0) -- (52,30);
	\draw[->] (-1,26.5) -- (53,26.5);
	\draw[->] (1,23) -- (51,23);
	\draw[->] (3,19.5) -- (49,19.5);
	\draw[->] (5,16) -- (47,16);
	\node[]	at (6,12) {\scriptsize  $Y_{j,-+}^{\cV'_n,\cV_m}(x/z)$};
	\node[]	at (46,12) {\scriptsize  $Y_{j,++}^{\cV_m,\cV'_n}(xz)$};
	\node[]	at (13,2) {\scriptsize $B_{j}^{\cV_{m}}(x)$};
	\node[] at (4,16) {\small $v_{1}$};
	\node[]	at (-2,26.5) {\small $v_n$}; 
	\node[] at (-.2,23) {\small $v_{n\mns1}$};
	\node[] at (2,19.5) {\small $v_{2}$};
	\node[right] at (-1.5,31) {\small $u_{1}\hspace{.7cm} u_{2}\hspace{.6cm} u_{m \mns 1}\hspace{.6cm} u_{m}$};
	\node[right] at (39.3,31) {\small $u_{1}\hspace{.7cm} u_{2}\hspace{.6cm} u_{m \mns 1}\hspace{.6cm} u_{m}$};
	\node[right] at (19,-1) {\small $u_{1}\hspace{.7cm} u_{2}\hspace{.6cm} u_{m \mns 1}\hspace{.6cm} u_{m}$};%
	\node[]	at (3,27.7) {\small $\frac{u_1}{v_n}$};
		\node[]	at (7,27.7) {\small $\frac{u_2}{v_n}$};
			\node[]	at (11,27.7) {\small $\frac{u_{m\mns1}}{v_n}$};
				\node[]	at (15,27.7) {\small $\frac{u_m}{v_n}$};
	\node[]	at (5.5,24.2) {\small $\frac{u_1}{v_{n\mns1}}$};
		\node[]	at (9.5,24.2) {\small $\frac{u_2}{v_{n\mns1}}$};
			\node[]	at (13.5,24.2) {\small $\frac{u_{m\mns1}}{v_{n\mns1}}$};
				\node[]	at (17.7,24.2) {\small $\frac{u_m}{v_{n\mns1}}$};
	\node[]	at (7.5,20.6) {\small $\frac{u_1}{v_2}$};
		\node[]	at (11.5,20.6) {\small $\frac{u_2}{v_2}$};
			\node[]	at (15.5,20.6) {\small $\frac{u_{m\mns1}}{v_2}$};
				\node[]	at (19.7,20.6) {\small $\frac{u_m}{v_2}$};
	\node[]	at (9.8,17.2) {\small $\frac{u_1}{v_1}$};
		\node[]	at (13.9,17.2) {\small $\frac{u_2}{v_1}$};
			\node[]	at (18.1,17.2) {\small $\frac{u_{m\mns1}}{v_1}$};
				\node[]	at (22,17.2) {\small $\frac{u_m}{v_1}$};
	\node[]	at (39.8,27.4) {\scriptsize $u_1 v_n$};
		\node[]	at (43.8,27.4) {\scriptsize $u_2 v_n$};
			\node[]	at (48.2,27.4) {\scriptsize $u_{m\mns1} v_n$};
				\node[]	at (52,27.4) {\scriptsize $u_m v_n$};
	\node[]	at (37.7,23.8) {\scriptsize $u_1 v_{n\mns1}$};
		\node[]	at (41.7,23.8) {\scriptsize $u_2 v_{n\mns1}$};
			\node[]	at (45.9,23.8) {\scriptsize $u_{m\mns1} v_{n\mns1}$};
				\node[]	at (49.8,23.8) {\scriptsize $u_m v_{n\mns1}$};
	\node[]	at (35.1,20.3) {\scriptsize $u_1 v_2$};
		\node[]	at (39,20.3) {\scriptsize $u_2 v_2$};
			\node[]	at (43.4,20.3) {\scriptsize $u_{m\mns1}v_2$};
				\node[]	at (47.2,20.3) {\scriptsize $u_m v_2$};
	\node[]	at (32.8,16.8) {\scriptsize $u_1 v_1$};
		\node[]	at (36.7,16.8) {\scriptsize $u_2 v_1$};
			\node[]	at (41,16.8) {\scriptsize $u_{m\mns1} v_1$};
				\node[]	at (44.9,16.8) {\scriptsize $u_m v_1$};
	\node[]	at (27.4,9) {\tiny $u_{1} u_{m}$};
		\node[]	at (25.8,5.9) {\tiny $u_{1} u_{m\mns1}$};
		\node[]	at (29.7,5.9) {\tiny $u_{2} u_{m}$};
			\node[]	at (23.3,2.9) {\tiny $u_{1} u_{2}$};
			\node[]	at (27.8,2.9) {\tiny $u_{2} u_{m\mns1}$};
			\node[]	at (32,2.9) {\tiny $u_{m\mns1} u_{m}$};
	\node[color=black!70]	at (1,25.7) {\tiny $j\pls n\mns2$};
		\node[color=black!70]	at (3.3,22.2) {\tiny $j\pls n\mns1$};
			\node[color=black!70]	at (6.1,18.7) {\tiny $j\pls1$};
				\node[color=black!70]	at (8.9,15.2) {\tiny $j$};
				\node[color=black!70]	at (12.3,15.2) {\tiny $j\pls1$};
				\node[color=black!70]	at (15.9,15.2) {\tiny $j\pls m\mns 2$};
				\node[color=black!70]	at (19.9,15.2) {\tiny $j\pls m\mns 1$};
	\node[color=black!70]	at (34.5,25.7) {\tiny $j\pls m\pls n\mns2$};
		\node[color=black!70]	at (32.2,22.2) {\tiny $j\pls m\pls n\mns3$};
			\node[color=black!70]	at (30.7,18.7) {\tiny $j\pls m$};
				\node[color=black!70]	at (28,15.2) {\tiny $j\pls m\mns1$};
				\node[color=black!70]	at (32.3,15.2) {\tiny $j\pls m\mns2$};
				\node[color=black!70]	at (36.7,15.2) {\tiny $j\pls1$};
				\node[color=black!70]	at (41.1,15.2) {\tiny $j$};
	\node[color=black!70]	at (24,9) {\tiny $j\pls m\mns2$};
		\node[color=black!70]	at (22.1,5.9) {\tiny $j\pls m\mns3$};
			\node[color=black!70]	at (21,2.9) {\tiny $j$};
\etpm 
\end{center}
\caption[Composite operator]{
{The composite operator {\small  $Y_{j,-+}^{\cV'_n,\cV_m}(x/z)\,B_j^{\cV_{m}}(x)\,Y_{j,++}^{\cV_m,\cV'_n}(xz)$}. The notion employed in the diagram is:
\medskip\\
{\small
\btp 
	\draw[-] (-4,0) -- (-4,0);  
	\draw[-] (-.5,.8) -- (.5,-.8) ;
	\draw[-] (-.8,0) -- (.8,0);
	\node[] at (1.7,-.2) {\small $\frac{u_i}{v_j}$};
\etp = $X_{k}(\frac{x\, u_i}{z\, v_j})$,\qquad
\btp 
	\draw[-] (.5,.8) -- (-.5,-.8) ;
	\draw[-] (-.8,0) -- (.8,0);
	\node[] at (2.2,-.2) {\small ${u_i}{v_j}$};
\etp = $X_{k}(x z\, u_i v_j)$,\qquad
\btp 
	\draw[-] (.5,.7) -- (-.5,-.7) ;
	\draw[-] (.5,-.7) -- (-.5,.7) ;
	\node[] at (1.8,-.2) {\small ${u_i}{u_j}$};
\etp = $X_{k}(x^2 u_i u_j)$,\qquad
\btp 
	\draw[-] (-.7,.8) -- (0,0) -- (.7,.8);
	\draw[-] (-.7,0) -- (.7,0);
	\node[] at (0,-.7) {\small $u_i$};
\etp \, =  $\tau_j(x u_i)$,}
\medskip\\
where $k=j,\,\ldots,\,j\!+\!m\!+\!n\!-\!2$ measures the distance to the boundary (horizontal bottom line). The dashed line represents the right factors of \eqref{Ymfac}, \eqref{Ypfac} and \eqref{RE:fact1}.}
} \label{fig:RE-YBY}
\end{figure}
\noindent {\it Step 1.} 
Consider the composite operator $\cX_{n,j+m-1}(x u_{m}/(z \cev{v})) \cX_{m-1,j}(x^2 u_m \vec{u}_\circ\!)  Y_{j+1,++}^{\cV_{m-1},\cV'_n}(xz)$ bound\-ed by the dashed line in figure \ref{fig:RE-YBY}. Observe that the inner triangles having the dashed line as their left edge form a family of fused YBEs \eqref{fYBE3} with quantum contents serving as the coordinates:
\begin{equation*}
\btp 
	\draw[-,dotted] (-1.2,.6) -- (0,-1.0);
	\draw[-] (0,-1.0) -- (1.2,.6) -- (-1.2,.6);
	\node[] at (2.6,-.2) {$\to$};	
	\draw[-,dotted] (6.4,-1) -- (5.2,.6);
	\draw[-] (6.4,-1) -- (4,-1) -- (5.2,.6);
\etp \quad : \quad \left( \frac{u_m}{\cev{v}} ,\, u_i u_m ,\, u_i \vec{v} \right).
\end{equation*} 
For example the innermost fused YBE is
\begin{align}
& \cX_{n,j+m-1}({x u_m}/({z \cev{v}})) \, X_{j+m-2}(x^2 u_1 u_m) \, \vX_{n,j+m-1}(xz \, u_1 \vec{v}) = \el
& \qquad \vX_{n,j+m-2}(xz \, u_1 \vec{v}) \, X_{j+m+n-2}(x^2 u_1 u_m) \, \cX_{n,j+m-2}({xu_m}/({z\cev{v}})) \,.
\end{align}
Thus, by a virtue of fused YBEs, we obtain
%
%
\begin{align} \label{RE:step1}
& \cX_{n,j+m-1}(x u_{m}/(z \cev{v}\,)) \, \cX_{m-1,j}(x^2 u_m \vec{u}_\circ) \, Y_{j+1,++}^{\cV_{m-1},\cV'_n}(xz) \el
& \qquad = Y_{j,++}^{\cV_{m-1},\cV'_n}(xz) \, \cX_{m-1,j+n}(x^2 u_m \vec{u}_\circ) \, \cX_{n,j}(x u_{m}/(z \cev{v}\,)) \,.
\end{align}

\begin{figure} \begin{center}\btpm
	\node[right] at (9,-1) {\small $u_{m}\hspace{1.45cm} v_{1}\hspace{1.5cm}  v_{2}\hspace{1.45cm} v_{n\mns1}\hspace{1.3cm} v_{n}$};
	\node[]	at (27.5,9) {\scriptsize $B_{j}^{\cV'_{n}}(z)$};
	\node[]	at (21.2,15.5) {\scriptsize $\vX_{n,j}(x z u_{m}\vec{v})$};
	\node[]	at (1,3) {\scriptsize $\cX_{n,j}(x u_{m}/(z \cev{v}))$};
	\draw[-] (0,0) -- (40,0);
	\draw[->,dashed] (1,21) -- (10,0) -- (19,21);
	\draw[->] (0,20) -- (38,0) -- (40,1);
	\draw[->] (0,16) -- (31,0) -- (40,4.5);
	\draw[->] (0,12) -- (24,0) -- (40,8);
	\draw[->] (0,8) -- (17,0) -- (40,11.5);
	\node[]	at (3.3,20) {\small $\frac{u_m}{v_n}$};
		\node[]	at (5.5,15) {\small $\frac{u_m}{v_{n\mns1}}$};
			\node[]	at (7.1,10.2) {\small $\frac{u_m}{v_2}$};
				\node[]	at (9.3,5.3) {\small $\frac{u_m}{v_1}$};
	\node[color=black!70]	at (0.2,18.3) {\tiny $j\pls n\mns1$};
		\node[color=black!70]	at (2.2,13.1) {\tiny $j\pls n\mns2$};
			\node[color=black!70]	at (5,8.1) {\tiny $j\pls1$};
				\node[color=black!70]	at (7.3,3.5) {\tiny $j$};
	\node[]	at (17.3,12.5) {\scriptsize $u_m v_n$};
		\node[]	at (16.4,9.2) {\scriptsize ${u_m}{v_{n\mns1}}$};
			\node[]	at (14.5,6.1) {\scriptsize ${u_m}{v_2}$};
				\node[]	at (13.2,3.1) {\scriptsize ${u_m}{v_1}$};
	\node[]	at (30.2,5.3) {\scriptsize $v_{1} v_{m}$};
		\node[]	at (27,3.5) {\scriptsize $v_{1} v_{m\mns1}$};
		\node[]	at (33.7,3.5) {\scriptsize $v_{2} v_{m}$};
			\node[]	at (22.8,1.7) {\scriptsize $v_{1} v_{2}$};
			\node[]	at (30.5,1.7) {\scriptsize $v_{2} v_{m\mns1}$};
			\node[]	at (37.6,1.7) {\scriptsize $v_{m\mns1} v_{m}$};
	\node[color=black!70]	at (25.1,5.3) {\tiny $j\pls n\mns2$};
		\node[color=black!70]	at (21.3,3.5) {\tiny $j\pls n\mns3$};
			\node[color=black!70]	at (18.5,1.7) {\tiny $j$};

\etpm
\end{center}
\caption[Composite operator]{The composite operator {\small $\cX_{n,j}(x u_{m}/(z \cev{v})) \, \tau_j(x\,u_m)\, \vX_{n,j}(x z u_{m}\vec{v}) \, B_{j}^{\cV'_{n}}(z)$}. The notion employed in the diagram is: 
\medskip\\
{\small\hspace{0.1cm}
\btp 
	\draw[] (-2,0) -- (-2,0) ; 
	\draw[-,dashed] (-.3,.8) -- (.3,-.8) ;
	\draw[-] (-.8,.3) -- (.8,-.3);
	\node[] at (1.7,-.2) {\small $\frac{u_i}{v_j}$};
\etp = $X_{k}(\frac{x\, u_i}{z\, v_j})$,\;\;
\btp 
	\draw[-,dashed] (.3,.8) -- (-.3,-.8) ;
	\draw[-] (-.8,.3) -- (.8,-.3);
	\node[] at (2.2,-.2) {\small ${u_i}{v_j}$};
\etp = $X_{k}(x z\, u_i v_j)$,\;\;
\btp 
	\draw[-] (.6,.6) -- (-.6,-.6) ;
	\draw[-] (.6,-.6) -- (-.6,.6) ;
	\node[] at (2,-.2) {\small ${v_i}{v_j}$};
\etp = $X_{k}(z^2 v_i v_j)$,\;\;
\btp 
	\draw[-,dashed] (-.8,.8) -- (0,0) -- (.8,.8);
	\draw[-] (-.8,0) -- (.8,0);
	\node[] at (0,-.7) {\small $u_m$};
\etp \, =  $\tau_j(x u_m)$, \;
\btp 
	\draw[-] (-.8,.8) -- (0,0) -- (.8,.8);
	\draw[-] (-.8,0) -- (.8,0);
	\node[] at (0,-.7) {\small $v_i$};
\etp \, =  $\tau_j(z v_i)$,
}
\medskip\\
where $k=j,\,\ldots,\,j\!+\!n\!-\!1$ measures the distance to the boundary (horizontal bottom line). } \label{fig:RE-XXB}
\end{figure}

\noindent{\it Step 2.} 
Consider the composite operator $\cX_{n,j}(x u_{m}/(z \cev{v}))\, \tau_j(x\,u_m)\, \vX_{n,j}(x z u_{m}\vec{v})\, B_{j}^{\cV'_{n}}(z)$ shown in figure \ref{fig:RE-XXB}. Observe that the inner triangles having the dashed line as their left edge form a family of REs and two families of fused YBEs \eqref{fYBE1} and \eqref{fYBE3} (that follow after applying REs, i.e.\ can be seen in the figure after shifting the dashed line right by an appropriate RE) with quantum contents serving as the coordinates:
\begin{equation*}
\btp 
	\draw[-,dotted] (0,.7) -- (.7,-1.2);
	\draw[-,dotted] (.7,-1.2) -- (1.3,.0);
	\draw[-] (.7,-1.2) -- (3.5,-1.2) -- (0,.7);
	\node[] at (4.8,-.2) {$\to$};
	\draw[-,dotted] (8.15,-.1) -- (8.8,-1.2);
	\draw[-,dotted] (8.8,-1.2) -- (9.5,.7);
	\draw[-] (8.8,-1.2) -- (6,-1.2) -- (9.5,.7);
\etp \;\; : \;\; \left( \frac{u_m}{v_i} ,\, u_m ,\, u_m v_i ,\,  v_i \right),\qquad
\end{equation*} 
and
\begin{equation*}
\btp 
	\draw[-,dotted] (-2,-1.1) -- (-1,1.2);
	\draw[-] (-1,1.2) -- (.3,.2) -- (-2,-1.1);
	\node[] at (2,-.2) {$\to$};
	\draw[-,dotted] (3.8,-.2) -- (5,-1.1) -- (6.1,1.2);
	\draw[-] (6.1,1.2) -- (3.8,-.2);
\etp \;\; : \;\; \left( \frac{u_m}{v_i} ,\, u_m \vec{v}_{\{i\}} ,\, v_i \vec{v}_{\{i\}} \right),
\qquad
\btp 
	\draw[-,dotted] (-1.5,1) -- (-1,-1.1);
	\draw[-] (-1,-1.1) -- (.8,.4) -- (-1.5,1);
	\node[] at (2.2,-.2) {$\to$};
	\draw[-,dotted] (5.3,1.1) -- (6,-1.1);
	\draw[-] (5.3,1.1) -- (3.8,-.6) -- (6,-1.1);
\etp \;\; : \;\; \left( \frac{u_m}{\cev{v}_{\{i\}}} ,\, u_m v_i ,\, v_i \vec{v}_{\{i\}} \right),
\end{equation*}
where $\vec{v}_{\{i\}}=(v_{i+1},v_{i+2},\ldots,v_n)$ and $\cev{v}_{\{i\}}$ is its reverse. For example the innermost RE is
\begin{equation}
X_{j}({x u_m}/({z v_1})) \,\tau_j(xu_m)\, X_{j}(xz u_m v_1) \, \tau_j(z v_1)=
\tau_j(z v_1) \, X_{j}(xz u_m v_1) \, \tau_j(xu_m) \, X_{j}({x u_m}/({z v_1})) \,.
\end{equation}
It follows by the fused YBEs,
\begin{align}
& X_1({x u_m}/({z v_1})) \, \vX_{n-1,2}(xz u_m \vec{v}_{\{1\}}) \, \vX_{n-1,1}(z^2 v_1 \vec{v}_{\{1\}}) 
\el
& \qquad = \vX_{n-1,2}(z^2 v_1 \vec{v}_{\{1\}}) \, \vX_{n-1,1}(xz u_m \vec{v}_{\{1\}}) \, X_n({x u_m}/({z v_1})),
\end{align}
and
\begin{align}
& \cX_{n-1,2}({x u_m}/({z \cev{v}_{\{1\}}})) \, X_1(xz u_m v_1) \, \vX_{n-1,2}(z^2 v_1 \vec{v}_{\{1\}}) \el
& \qquad = \vX_{n-1,1}(z^2 v_1 \vec{v}_{\{1\}}) \, X_n(xz u_m v_1) \, \cX_{n-1,1}({x u_m}/({z \cev{v}_{\{1\}}})) .
\end{align}
These three steps are equivalent to shifting leftwards the line associated with $v_1$ in figure \ref{fig:RE-XXB}. Then, by repeating these steps for $n-1$ times we obtain
%
%
\begin{equation} \label{RE:step2}
\cX_{n,j}(x u_{m}/(z \cev{v}))\, \tau_j(x\,u_m)\, \vX_{n,j}(x z u_{m}\vec{v}) \, B_{j}^{\cV'_{n}}(z) = B_{j}^{\cV'_{n}}(z)\, \cX_{n,j}(x z u_{m}\vec{v})\, \tau_j(x\,u_m)\, \vX_{n,j}(x u_{m}/(z \cev{v})) \,.
\end{equation}
This {\it step} is a special case of the fused RE which is obtained when one of the boundary operators is primitive, namely $\tau_j(x u_m)$.

\noindent{\it Step 3.} 
By combining the RHS of \eqref{RE:step0}, \eqref{RE:step1} and \eqref{RE:step2} we obtain
%
%
\begin{align} \label{RE:step3}
& Y_{j,-+}^{\cV'_n,\cV_{m-1}} (x/z)\, B_j^{\cV_{m-1}}(x)\, Y_{j,++}^{\cV_{1},\cV'_n}(xz)\,B_{j}^{\cV'_{n}}(z) \el
& \qquad\qquad \times  \cX_{m-1,j+n}(x^2 u_m \vec{u}_\circ) \, \cX_{n,j}(x z u_{m}\vec{v})\, \tau_j(x\,u_m)\, \vX_{n,j}(x u_{m}/(z \cev{v})) \el
& \quad = B_{j}^{\cV'_{n}}(z) \, Y_{j,-+}^{\cV'_n,\cV_{m-1}} (xz)\, B_j^{\cV_{m-1}}(x) \, Y_{j,+-}^{\cV_{m-1},\cV'_n}(x/z) \el
& \qquad\qquad \times  \cX_{m-1,j+n}(x^2 u_m \vec{u}_\circ) \, \cX_{n,j}(x z u_{m}\vec{v})\, \tau_j(x\,u_m)\, \vX_{n,j}(x u_{m}/(z \cev{v})) \,,
\end{align}
where we have used the induction hypothesis.

\begin{figure}
\begin{center}\btpm
	\node[]	at (2,13) {\scriptsize  $Y_{j,+-}^{\cV_{m-1},\cV'_n}(x/z)$};
	\node[]	at (27,23) {\scriptsize $\cX_{m-1,j+n}(x^2 u_m \vec{u}_\circ)$};
	\node[]	at (32,13) {\scriptsize $\cX_{n,j}(x z u_{m}\vec{v})$};
	\draw[->] (4,9) -- (30,9);
	\draw[->] (2,6) -- (32,6);
	\draw[->] (0,3) -- (34,3);
	\draw[->] (-2,0) -- (36,0);
	\draw[->] (-1,-1) -- (17,24);
	\draw[->] (4,-1) -- (19.5,21);
	\draw[->] (9,-1) -- (22,18);
	\draw[->,dashed] (14.5,23) -- (33,-1);
	\node[right,color=black!70] at (10.3,21.8) {\tiny $j\pls n\pls m\mns2$};
		\node[right,color=black!70] at (12.9,18.6) {\tiny $j\pls n\pls m\mns3$};
			\node[right,color=black!70] at (17.5,15.5) {\tiny $j\pls n$};
	\node[right,color=black!70] at (22.1,8) {\tiny $j\pls n\mns1$};
		\node[right,color=black!70] at (24.3,5) {\tiny $j\pls n\mns2$};
			\node[right,color=black!70] at (27.7,2) {\tiny $j\pls1$};
				\node[right,color=black!70] at (30.9,-1) {\tiny $j$};
	\node[right] at (15.8,21.8) {\scriptsize $u_m u_1$};
		\node[right] at (18,18.8) {\scriptsize $u_m u_2$};
			\node[right] at (20.5,15.5) {\scriptsize $u_m u_{m\mns1}$};
	\node[right] at (1.7,9) {\small $v_1$};
		\node[right] at (-.5,6) {\small $v_2$};
			\node[right] at (-3,3) {\small $v_{n\mns1}$};
				\node[right] at (-4.4,0) {\small $v_n$};
	\node[right,color=black!70] at (-.6,-.7) {\tiny $j\pls m\mns2\hspace{.6cm} j\pls1\hspace{1cm} j$};
	\node[right] at (-2.5,-2) {\small $u_{1}\hspace{1cm} u_{2}\hspace{1cm} u_{m \mns 1}\hspace{5.8cm} u_{m}$};
	\node[right] at (7.3,10.2) {\small $\frac{u_1}{v_1} \hspace{1cm} \frac{u_2}{v_1} \hspace{1cm} \frac{u_{m\mns1}}{v_1} \hspace{1.6cm} $\scriptsize ${u_m}{v_1}$};
	\node[right] at (5.1,7.2) {\small $\frac{u_1}{v_2} \hspace{1cm} \frac{u_2}{v_2} \hspace{1cm} \frac{u_{m\mns1}}{v_2} \hspace{2.9cm} $\scriptsize ${u_m}{v_2}$};
	\node[right] at (2.8,4.2) {\small $\frac{u_1}{v_{n\mns1}} \hspace{.8cm} \frac{u_2}{v_{n\mns1}} \hspace{.85cm} \frac{u_{m\mns1}}{v_{n\mns1}}  \hspace{4.1cm} $\scriptsize ${u_m}{v_{n\mns1}}$};
	\node[right] at (.85,1.2) {\small $\frac{u_1}{v_n} \hspace{1cm} \frac{u_2}{v_n} \hspace{1.05cm} \frac{u_{m\mns1}}{v_n}  \hspace{5.4cm} $\scriptsize ${u_m}{v_n}$};

\etpm 
\end{center}
\caption[Composite operator]{The composite operator {\small  $Y_{j,+-}^{\cV_{m-1},\cV'_n}(x/z)\, \cX_{m-1,j+n}(x^2 u_m \vec{u}_\circ)\, \cX_{n,j}(x z u_{m}\vec{v})$}. The notion employed in the diagram is: 
\medskip\\
\hspace{2.6cm}{\small
\btp 
	\draw[-] (-10,0) -- (-10,0) ; 
	\draw[-] (-.5,-.8) -- (.5,.8) ;
	\draw[-] (-.8,0) -- (.8,0);
	\node[] at (1.7,-.2) {\small $\frac{u_i}{v_j}$};
\etp = $X_{k}(\frac{x\, u_i}{z\, v_j})$,\qquad
\btp 
	\draw[-] (.5,.7) -- (-.5,-.7) ;
	\draw[-,dashed] (.5,-.7) -- (-.5,.7) ;
	\node[] at (1.8,-.2) {\small ${u_i}{u_j}$};
\etp = $X_{k}(x^2 u_i u_j)$,\qquad
\btp 
	\draw[-,dashed] (.5,-.8) -- (-.5,.8) ;
	\draw[-] (-.8,0) -- (.8,0);
	\node[] at (2.2,-.2) {\small ${u_i}{v_j}$};
\etp = $X_{k}(x z\, u_i v_j)$,
}
\medskip\\
where $k=j,\,\ldots,\,j\!+\!m\!+\!n\!-\!2$ measures the distance to the boundary (not present in this diagram). 
} \label{fig:RE-YYY}
\end{figure}

\noindent{\it Step 4.} 
Consider the composite operator {\small $Y_{j,+-}^{\cV_{m-1},\cV'_n}(x/z)\, \cX_{m-1,j+n}(x^2 u_m \vec{u}_\circ)\, \cX_{n,j}(x z u_{m}\vec{v})$} shown in figure \ref{fig:RE-YYY}. This operator encloses a family of fused YBEs \eqref{fYBE3} via the inner triangles having the dashed line as their right edge and quantum contents serving as their coordinates:
\begin{equation*}
\btp 
	\draw[-,dotted] (0,0.6) -- (1,-1.1);
	\draw[-] (0,0.6) -- (-1,-1.1) -- (1,-1.1);
	\node[] at (2.2,-.2) {$\to$};	
	\draw[-,dotted] (3.6,.6) -- (4.6,-1.1);
	\draw[-] (3.6,.6) -- (5.6,.6) -- (4.6,-1.1);
\etp \;\; : \;\; \left( \frac{u_i}{\cev{v}} ,\, u_i u_m ,\, u_m \vec{v} \right).
\end{equation*} 
For example the innermost fused YBE is
\begin{align}
& \vX_{n,j}(\frac{x u_{m-1}}{z \cev{v}})\, X_{j+n}(x^2  u_m u_{m-1}) \, \cX_{n,j}(xz u_m \vec{v}) \el
& \qquad = \cX_{n,j+1}(xz u_m \vec{v}) \, X_{j}(x^2 u_m u_{m-1}) \, \vX_{j+n-1}(\frac{x u_{m-1}}{z \cev{v}}) .
\end{align}
Thus, by a virtue of fused YBEs, we obtain
%
%
\begin{align} \label{RE:step4}
& Y_{j,+-}^{\cV_{m-1},\cV'_n}(x/z)\, \cX_{m-1,j+n}(x^2 u_m \vec{u}_\circ)\, \cX_{n,j}(x z u_{m}\vec{v}) \el
& \qquad = \cX_{n,j+m-1}(x z u_m \vec{v}) \, \cX_{m-1,j}(x^2 u_m \vec{u}_\circ) \,Y_{j+1,+-}^{\cV_{m-1},\cV'_n}(x/z) \,.
\end{align}

\noindent Finally, by putting all the steps together, we have
\begin{align}
& Y_{j,-+}^{\cV'_n,\cV_m}(x/z)\,B_j^{\cV_{m}}(x)\,Y_{j,++}^{\cV_m,\cV'_n}(xz) B_{j}^{\cV'_{n}}(z) \el
& \quad \overset{\eqref{RE:step0}}{=}  Y_{j,-+}^{\cV'_n,\cV_{m-1}}(x/z) \, B_j^{\cV_{m-1}}(x) \, \cX_{n,j+m-1}(x u_{m}/(z \cev{v}\,)) \, \cX_{m-1,j}(x^2 u_m \vec{u}_\circ) \el
& \hspace{1.3cm} \times  Y_{j+1,++}^{\cV_{m-1},\cV'_n}(xz) \, \tau_j(x\,u_m) \, \vX_{n,j}(xz u_{m}\vec{v}) \, B_j^{\cV'_{n}}(z) \el
& \quad \overset{\eqref{RE:step1}}{=} 
Y_{j,-+}^{\cV'_n,\cV_{m-1}}(x/z)\,B_j^{\cV_{m-1}}(x)\,Y_{j,++}^{\cV_{m-1},\cV'_n}(xz) \el
& \hspace{1.3cm} \times \cX_{m-1,j+n}(x^2 u_m \vec{u}_\circ)\, \cX_{n,j}(xu_{m}/(z\cev{v})) \, \tau_j(x\,u_m) \, \vX_{j}(x z u_{m}\vec{v}) \, B_{j}^{\cV'_{n}}(z)
\el
& \quad \overset{\eqref{RE:step2}}{=}
Y_{j,-+}^{\cV'_n,\cV_{m-1}}(x/z)\,B_j^{\cV_{m-1}}(x)\,Y_{j,++}^{\cV_{m-1},\cV'_n}(xz) \, \cX_{m-1,j+n}(x^2 u_m \vec{u}_\circ) \el
& \hspace{1.3cm} \times B_{j}^{\cV'_{n}}(z)\, \cX_{n,j}(x z u_{m}\vec{v})\, \tau_j(x\,u_m)\, \vX_{n,j}(x u_{m}/(z \cev{v}))
\el
& \quad 
\overset{\eqref{RE:step3}}{=}
B_{j}^{\cV'_{n}}(z) \, Y_{j,++}^{\cV'_n,\cV_{m-1}} (xz)\, B_j^{\cV_{m-1}}(x) \, Y_{j,+-}^{\cV_{m-1},\cV'_n}\!(x/z) \, \cX_{m-1,j+n}(x^2 u_m \vec{u}_\circ) \el
& \hspace{1.3cm} \times \cX_{n,j}(x z u_{m}\vec{v})\, \tau_j(x\,u_m)\, \vX_{n,j}(x u_{m}/(z \cev{v}))
\el
& \quad 
\overset{\eqref{RE:step4}}{=}
B_{j}^{\cV'_{n}}(z) \, Y_{j,++}^{\cV'_n,\cV_{m-1}} (xz)\, B_j^{\cV_{m-1}}(x) \, \cX_{n,j+m-1}(x z u_m \vec{v}) \, \cX_{m-1,j}(x^2 u_m \vec{u}_\circ) \,Y_{j+1,+-}^{\cV_{m-1},\cV'_n}(x/z)
\el
& \hspace{1.3cm} \times \tau_j(x\,u_m)\, \vX_{n,j}(x u_{m}/(z \cev{v})) 
\el
& \quad
\overset{\eqref{RE:step5}}{=} B_{j}^{\cV'_{n}}(z) \, Y_{j,++}^{\cV'_n,\cV_m} (xz) B_j^{\cV_{m}}(x) \, Y_{j,+-}^{\cV_m,\cV'_n}(x/z) \,,
\end{align}
as required.
\end{proof} 


\noindent We will now give a specialization of the fused RE for a pair of tableaux $\cT_m$ and $\cT'_n$.
\begin{thrm} \label{th42} 
Fused projected boundary operators $(\forall x,z;\; \forall \cT_m, \cT'_n)$
\begin{equation}
\bB_j^{\cT_{m}}(x) = A_j^{\cT_{m}} B_j^{\cT_{m}}(x), \label{FusedBA} 
\end{equation}
are solutions of the fused projected reflection equation 
\begin{equation} \label{fREA}
\bY_{j,-+}^{\cT'_n,\cT_m}(x/z)\,\bB_j^{\cT_{m}}(x)\,\bY_{j,++}^{\cT_m,\cT'_n}(xz)\,\bB_j^{\cT'_{n}}(z) = \bB_j^{\cT'_{n}}(z)\,\bY_{j,++}^{\cT'_n,\cT_m}(xz)\,\bB_j^{\cT_{m}}(x,\xi)\,\bY_{j,+-}^{\cT_m,\cT'_n}(x/z) \,,
\end{equation}
where
\begin{align}
 \bY_{j,-+}^{\cT_m,\cT'_n}(x/z) &=  A_{j}^{\cT_m} A_{j+m}^{\cT'_n}\, Y_{j,-+}^{\cT_m,\cT'_n}(x/z) \,,\label{bY1} \\
 \bY_{j,++}^{\cT_m,\cT'_n}(xz) &=  A_{j+m}^{\cT'_n} \, Y_{j,++}^{\cT_m,\cT'_n}(xz) \,, \label{bY2}\\
 \bY_{j,+-}^{\cT_m,\cT'_n}(x/z) &=  Y_{j,+-}^{\cT_m,\cT'_n}(x/z) \,. \label{bY3}
\end{align}
\end{thrm}
\noindent For further reference we define the following operator
\begin{equation}
\bY_{j,--}^{\cT'_n,\cT_m}(b/z) = A^{\cT_m}_j \, Y_{j,--}^{\cT'_n,\cT_m}(b/z) \,.\label{bY4}
\end{equation}
Next we give a preparation Lemma required for the proof of the Theorem above.
\begin{lemma} \label{lemma41} 
In the Hecke algebra the following identities hold $(\forall x,\;\forall\cT_m,\cT'_n)$
\begin{align}
A_{j}^{\cT_m} A_{j+m}^{\cT'_n} \, Y_{j,-+}^{\cT_m,\cT'_n}(x) \, A_{j}^{\cT'_{n}} A_{j+n}^{\cT_m} &= A_{j}^{\cT_m} A_{j+m}^{\cT'_n}\, Y_{j,-+}^{\cT_m,\cT'_n}(x) \,, \label{lemma41:1} \\
A_{j+m}^{\cT'_n} \, Y_{j,++}^{\cT_m,\cT'_n}(x) \, A_{j}^{\cT'_{n}} &= A_{j+m}^{\cT'_n} \, Y_{j,++}^{\cT_m,\cT'_n}(x) \,, \label{lemma41:2} \\
A_{j}^{\cT'_n} \, Y_{j,--}^{\cT_m,\cT'_n}(x) \, A_{j+m}^{\cT'_{n}} &= A_{j}^{\cT'_n} \, Y_{j,--}^{\cT_m,\cT'_n}(x) \,. \label{lemma41:3}
\end{align}
\end{lemma}
\begin{proof} 
This Lemma is a generalization of Lemma \ref{lemma31}. Consider \eqref{lemma41:1}. First, by \eqref{lemma31:1} and the first part of \eqref{Ymp} we have
\begin{equation}
A_{j+m}^{\cT'_n} \, Y_{j,-+}^{\cT_m,\cT'_n}(x) \, A_{j}^{\cT'_{n}} = A_{j+m}^{\cT'_n} \prod_{i=1\ldots m}^{\leftarrow} \vX_{n,j+i-1}(x\vec{e}/c_i) \, A_{j}^{\cT'_{n}} = A_{j+m}^{\cT'_n} \, Y_{j,-+}^{\cT_m,\cT'_n}(x) .
\end{equation}
Then, by \eqref{lemma31:2} and the second part of \eqref{Ymp} we have
\begin{equation}
A_{j}^{\cT_m} \, Y_{j,-+}^{\cT_m,\cT'_n}(x) \, A_{j+n}^{\cT_m} = A_{j}^{\cT_m} \prod_{i=1\ldots n}^{\rightarrow} \cX_{m,j+i-1}(x e_i/\cev{c}) \, A_{j+n}^{\cT_m} = A_{j}^{\cT_m} \, Y_{j,-+}^{\cT_m,\cT'_n}(x) .
\end{equation}
This gives \eqref{lemma41:1} as required. For \eqref{lemma41:2} we have
\begin{align}
 A_{j+m}^{\cT'_n} \, Y_{j,++}^{\cT_m,\cT'_n}(x) \, A_{j}^{\cT'_{n}} = A_{j+m}^{\cT'_n} \, \prod_{i=1\ldots m}^{\leftarrow} \vX_{n,j+i-1}(xc_{m-i+1}\vec{e}\,) \, A_{j}^{\cT'_{n}} = A_{j+m}^{\cT'_n} \, \bY_{j,++}^{\cT_m,\cT'_n}(x) \,,
\end{align}
where we have used \eqref{lemma31:1} and the first part of \eqref{Ypp}. Finally, for \eqref{lemma41:3} we have
\begin{align}
 A_{j}^{\cT'_n} \, Y_{j,--}^{\cT_m,\cT'_n}(x) \, A_{j+m}^{\cT'_{n}} = A_{j}^{\cT'_n} \, \prod_{i=1\ldots n}^{\rightarrow} \cX_{m,j+i-1}(x/(\cev{c}\,e_{n-i+1})) \, A_{j+m}^{\cT'_{n}} = A_{j}^{\cT'_n} \, Y_{j,--}^{\cT_m,\cT'_n}(x) \,,
\end{align}
where we have used \eqref{lemma31:2} and the first part of \eqref{Ymm}.
\end{proof}
Now we are ready to give a proof of Theorem \eqref{th42}.
\begin{proof}
The LHS of \eqref{fREA} can be written as
\begin{align}
& \bY_{j,-+}^{\cT'_n,\cT_m}(x/z)\,\bB_j^{\cT_{m}}(x)\,\bY_{j,++}^{\cT_m,\cT'_n}(xz)\,\bB_j^{\cT'_{n}}(z) \el
& \qquad = A_{j}^{\cT_n} A_{j+n}^{\cT_m} \, Y_{j,-+}^{\cT'_n,\cT_m}(x/z)\, A_{j}^{\cT_m} \, B_j^{\cT_{m}}(x)\, A_{j+m}^{\cT'_n} \, Y_{j,++}^{\cT_m,\cT'_n}(xz)\, A_{j}^{\cT'_n} \, B_j^{\cT'_{n}}(z) \el
& \qquad = A_{j}^{\cT_n} A_{j+n}^{\cT_m} \, Y_{j,-+}^{\cT'_n,\cT_m}(x/z)\, B_j^{\cT_{m}}(x)\, Y_{j,++}^{\cT_m,\cT'_n}(xz)\, B_j^{\cT'_{n}}(z) \,,
\end{align}
where we have used Lemma \ref{lemma41}. Then, by fused RE \eqref{fRE}, we have
\begin{align}
(4.36) &= A_{j}^{\cT_n} A_{j+n}^{\cT_m} \, B_j^{\cT'_{n}}(z)\, Y_{j,++}^{\cT'_n,\cT_m}(xz)\, B_j^{\cT_{m}}(x)\,Y_{j,+-}^{\cT_m,\cT'_n}(x/z) \el
& = A_{j}^{\cT_n} \, B_j^{\cT'_{n}}(z)\, A_{j+n}^{\cT_m} \, Y_{j,++}^{\cT'_n,\cT_m}(xz)\, A_{j}^{\cT_m} \, B_j^{\cT_{m}}(x)\,Y_{j,+-}^{\cT_m,\cT'_n}(x/z) \el
& = \bB_j^{\cT'_{n}}(z)\,\bY_{j,++}^{\cT'_n,\cT_m}(xz)\,\bB_j^{\cT_{m}}(x,\xi)\,\bY_{j,+-}^{\cT_m,\cT'_n}(x/z) \,,
\end{align}
which gives the RHS of the fused projected RE \eqref{fREA} as required.
\end{proof}


To conclude this section we give a Lemma stating an important property of the fused boundary operators. We will use this Lemma in the section below. 

\begin{lemma} \label{lemma42} 
In the affine Hecke algebra the following identity holds $(\forall x,\forall\cV_m)$
\begin{equation} \label{B_property}
\Phi_{m,j}(\vec{u})\, B^{\cV_{m}}_j(x) = B^{\cV_{m}^R}_j(x)\, \Psi_{m,j}(\vec{u}) \,,
\end{equation}
where $B^{\cV_{m}^R}_j(x)$ has reversed quantum contents with respect to $B^{\cV_{m}}_j(x)$.
\end{lemma}
\begin{proof} 
We will proof this Lemma by induction. The $m\!=\!1$ case gives the RE \eqref{RE}, thus is tautological. Suppose \eqref{B_property} holds for $\cV_{m-1}\subset\cV_m$. We factorize the LHS of \eqref{B_property} as
\begin{align}  \label{L42P:1}
\Phi_{m,j}(\vec{u})\, B^{\cV_{m}}_j(x) &= \vec{X}_{m-1,j}(\cev{u}_\circ/u_m)\, \Phi_{m-1,j}(\vec{u}_\circ)\, B^{\cV_{m-1}}_j(x) \, \cev{X}_{m-1,j}(x^2 \vec{u}_\circ u_m)\, \tau_j(x u_m) \el
&= \vec{X}_{m-1,j}(\cev{u}_\circ/u_m)\, B^{\cV_{m-1}^R}_j(x) \, \Psi_{m-1,j}(\vec{u}_\circ)\, \cev{X}_{m-1,j}(x^2 \vec{u}_\circ u_m)\, \tau_j(x u_m) \,,
\end{align}
where in the second equality we have used the induction hypothesis. Next, by a virtue of YBEs we have
\begin{equation}  \label{L42P:2}
\Psi_{m-1,j}(\vec{u}_\circ)\, \cev{X}_{m-1,j}(x^2 \vec{u}_\circ u_m) = \cev{X}_{m-1,j}(x^2 \cev{u}_\circ u_m)\,\Psi_{m-1,j+1}(\vec{u}_\circ) \,.
\end{equation}
In such a way the RHS of \eqref{L42P:1} becomes a fused RE \eqref{fRE} written in the reversed order,
\begin{align}
& \vec{X}_{m-1,j}(\cev{u}_\circ/u_m)\, B^{\cV_{m-1}^R}_j(x) \, \cev{X}_{m-1,j}(x^2 \cev{u}_\circ u_m)\, \tau_j(x u_m)\, \Psi_{m-1,j+1}(\vec{u}_\circ)\el
& \qquad  =  \tau_j(x u_m)\,\vec{X}_{m-1,j}(x^2\cev{u}_\circ u_m)\, B^{\cV_{m-1}^R}_j(x) \, \cev{X}_{m-1,j}(\cev{u}_\circ/u_m)\, \Psi_{m-1,j+1}(\vec{u}_\circ) \el
& \qquad  =  B^{\cV_{m}^R}_j(x) \, \Psi_{m,j}(\vec{u}) \,.
\end{align}
Here in the last equality we have used the identity
\begin{equation}
\cev{X}_{m-1,j}(\cev{u}_\circ/u_m)\, \Psi_{m-1,j+1}(\vec{u}_\circ) = \Psi_{m,j}(\vec{u}) \,,
\end{equation}
which follows by virtue of YBEs.
\end{proof}

\section{Transfer matrices and their functional relations}

This section contains the main results of this work. Here we construct a family of commutative elements, transfer matrices, which are analogues of the Sklyanin-type transfer matrices for the affine Hecke algebra. We will start by stating the necessary relations, required for the definition of the transfer matrices, and showing their commutativity property. We will then give recurrence and functional relations for two distinct families of Young tableaux, and demonstrate these relations and their solutions on some simple examples.

The construction of transfer matrix type operators requires a trace procedure, which in this case is the Ocneanu-Markov trace (see e.g. \cite{IsaevRev}). Let us recall its definition. Consider the following inclusions of the subalgebras $\chH_1 \subset \chH_2 \subset \dots \subset \chH_n$,
\begin{equation}
\{\tau_1; \sigma_1, \dots ,\sigma_{n-1}\} = \chH_n \subset \chH_{n+1} =
\{\tau_1; \sigma_1, \dots ,\sigma_{n-1},\sigma_n \} .
\end{equation}
The algebra $\chH_{n+1}$ can be equipped with a linear map $\Tr_{n+1}:\chH_{n+1} \to \chH_{n}$ to its natural subalgebra $\chH_{n}$,
such that for all $X,X' \in \chH_{n}$ and $Y \in \chH_{n+1}$ the following relations hold:
\begin{align}
& \Tr_{n+1} ( X ) = Q_0 \, X  ,\quad
\Tr_{n+1} ( X \, Y \, X' ) = X \, \Tr_{n+1}(Y) \, X' , \\
& \Tr_{n+1} ( \sigma_n^{\pm 1} X \sigma_n^{\mp 1}) = \Tr_{n} (X)  ,\quad
\Tr_{n+1} (X \sigma^{\pm1}_n X') =  Q_\pm \, X \, X' , \\
& \Tr_{1} (\tau_1^l)= Q_\tau^{(l)} ,\quad \Tr_{n}  \Tr_{n+1} ( \sigma_n Y ) = \Tr_{n} \Tr_{n+1} ( Y \sigma_n) , 
\end{align}
where $Q_\tau	^{(l)} \in \CCx$ $(l \in {\mathbb Z})$ are arbitrary constants, while $Q_{(0,\pm)}$ are constrained by $Q_+-Q_-=(q-q^{-1})Q_{0}$ which follows from \eqref{sigma_cycle}. The composition of the maps
\begin{align}
\mathcal{T}\!r_n = \Tr_1\Tr_2\cdots\Tr_n , \qquad \mathcal{T}\!r_n: \chH_{n}\to\CC,
\end{align}
is the Ocneanu--Markov trace. In what follows we will be heavily using the following notation
\begin{equation}
\dl \ldots \dr_{j} = \Tr_{j}, \qquad\quad \dl \ldots \dr^{j+m}_{j} = \Tr_{j} \Tr_{j+1} \cdots \Tr_{j+m-1} \Tr_{j+m}	.
\end{equation}

We will further give two propositions stating generalized cross-unitarity identities. These identities will further be used in constructing transfer matrices and showing their commutativity.

\begin{prop} \label{prop51} 
The following identity (called cross-unitarity in \cite{Isaev05}) holds $(\forall Y_k \in {\chH}_k, \forall x)$
\begin{equation} \label{CU}
\dl X_k(x)\, Y_k\, X_k(b/x) \dr_{k+1} = \dl Y_k \dr_k \,, \qquad\text{for}\qquad b=Q_+/Q_- \,.
\end{equation}
\end{prop}
\begin{proof} 
Consider the following element
\begin{equation}
\dl X_k(x)\, Y_k\, X_k(z) \dr_{k+1} \,.
\end{equation}
By definition \eqref{Xbax}, the expression above is equal to
\begin{align} \label{traceXYX}
\DL \bigg(\sigma_k+\frac{x \lambda}{1-x}\bigg)\, Y_k \,\bigg(\sigma^{-1}_k+\frac{\lambda}{1-z} \bigg) \DR_{k+1} &= \dl Y_k \dr_k + \lambda \DL Y_k \,\frac{x \lambda + (1-x) \sigma_k+x(1-z) \sigma^{-1}_k}{(1-x) (1-z)} \DR_{k+1} \el
& = \dl Y_k \dr_k + \lambda \, \frac{Q_{+}-Q_- x z}{(1-x) (1-z)} \, Y_k \,,
\end{align}
which upon $z=b/x$ and $b=Q_+/Q_-$ gives \eqref{CU} as required.
\end{proof} 
\begin{prop} \label{prop52} 
The following identity for the fused boundary operators $B_j^{\cV'_n}(x)$ holds \\ $(\forall x,z;\,\forall \cV_m, \cV'_n)$
\begin{equation} \label{CUB}
\dl Y_{j,++}^{\cV_m,\cV'_n}(z) \, B_j^{\cV'_{n}}(x)\, 
Y_{j,--}^{\cV'_n,\cV_m}(b/z) \dr^{j+m+n-1}_{j+m} = \dl B_j^{\cV'_{n}}(x) \dr^{j+n-1}_{j} .
\end{equation}
\end{prop}
\begin{proof}
We factorize the fused Yang-Baxter operators and use Proposition \ref{prop51}. This gives
\begin{align}
& \dl \vX_{n,j+m-1}(z \vec{v} u_1) \vX_{n,j+m-2}(z \vec{v} u_2) \cdots \vX_{j}(z \vec{v} u_m) \, B_j^{\cV'_{n}}(x) \el 
& \qquad \times \cX_{n,j}(b/(z\cev{v}u_m)) \cdots \cX_{n,j+m-2}(b/(z\cev{v}u_2)) \cX_{n,j+m-1}(b/(z\cev{v}u_1)) \dr^{j+m+n-1}_{j+m} = \el
& = \dl \vX_{n,j+m-2}(z \vec{v} u_2) \vX_{n,j+m-3}(z \vec{v} u_3) \cdots \vX_{j}(z \vec{v} u_m) \, B_j^{\cV'_{n}}(x) \el 
& \qquad \times \cX_{n,j}(b/(z\cev{v}u_m)) \cdots \cX_{n,j+m-3}(b/(z\cev{v}u_3)) \cX_{n,j+m-2}(b/(z\cev{v}u_2)) \dr^{j+m+n-2}_{j+m-1} .
\end{align}
Repeating this step for $m$ times we obtain $\dl B_j^{\cV'_{n}}(x) \dr^{j+n-1}_{j}$ as required.
\end{proof}
\begin{crl} \label{crl51} 
For the fused projected boundary operator $\bB_j^{\cT_{m}}(x)$ the identity \eqref{CUB} becomes
\begin{equation} \label{CUBA}
\dl \bY_{j,++}^{\cT_m,\cT'_n}(z) \, \bB_j^{\cT'_{n}}(x)\, 
\bY_{j,--}^{\cT'_n,\cT_m}(b/z) \dr^{j+m+n-1}_{j+m} = \dl \bB_j^{\cT'_{n}}(x) \dr^{j+n-1}_{j} ,
\end{equation}
which follows by application of Lemma \ref{lemma41} and properties of the trace.
\end{crl}

\smallskip


We are now ready to introduce transfer matrices in the affine Hecke algebra.
\begin{prop} \label{prop53}  
The elements
\begin{equation} \label{TM}
t_j^{\cV_{m}}(x) = \dl B^{\cV_{m}}_j(x) \dr^{j+m-1}_{j},
\end{equation} 
form a commutative family, $\big[t_j^{\cV_{m}}(x),t_j^{\cV'_{n}}(z)\big]=0$,\; $\forall\,x,z $ and $\forall \,\cV_m,\cV'_n$.
\end{prop}
\begin{proof} 
By application of Proposition \ref{prop52}, properties of the trace together with locality of $\tau_j(x)$, we have
\begin{align}
t_j^{\cV_{m}}(x)\,t_j^{\cV'_{n}}(z) &= \dl B^{\cV_{m}}_j(x)\, \dr^{j+m-1}_{j} \dl B^{\cV'_{n}}_j(z)\, \dr^{j+n-1}_{j}
\el
& = \dl B^{\cV_{m}}_j(x)\, Y^{\cV_m,\cV'_n}_{j,++}(xz)\, B^{\cV'_{n}}_j(z)\, Y^{\cV'_n,\cV_m}_{j,--}(b/(xz)) \dr^{j+m+n-1}_{j} \el
& = \dl Y^{\cV_m,\cV'_n}_{j,++}(xz)\, B^{\cV'_{n}}_j(z)\, Y^{\cV'_n,\cV_m}_{j,--}(b/(xz)) \, B^{\cV_{m}}_j(x) \dr^{j+m+n-1}_{j} \el
& = t_j^{\cV'_{n}}(z)\, t_j^{\cV_{m}}(x) \,,
\end{align}
as required. 
\end{proof}
\begin{crl} \label{crl53}  
The elements
\begin{equation} \label{TMA}
t_j^{\cT_{m}}(x) = \dl \bB^{\cT_{m}}_j(x) \dr^{j+m-1}_{j},
\end{equation} 
form a commutative family, $\big[t_j^{\cT_{m}}(x),t_j^{\cT'_{n}}(z)\big]=0$, $\forall\,x,z$ and $\forall \,\cT_m,\cT'_n$. This follows by Corollary \ref{crl51}. 
\end{crl}

Next we formulate a theorem stating a formal recurrence relation for the transfer matrices corresponding to the totally symmetric and antisymmetric tableaux. Below we will give an example of the explicit form of this recurrence relation for transfer matrices of small size. We will then derive a functional relation for the totally symmetric (resp.\ antisymmetric) and the hook--symmetric (resp.\ hook--antisymmetric) transfer matrices. We will discuss these operators in a greater detail for the free boundary model in the subsection below. 

\begin{thrm} \label{th52} 
Let $\cT_m=m_+$ (resp.\ $m_-$) be a totally symmetric (resp.\ totally antisymmetric) tableau of size $m\geq 2$. Then the corresponding transfer matrix $t^{m_\pm}_{j}(x)$ satisfies a formal recurrence relation 
\begin{equation} \label{TMrec}
t^{m_\pm}_{j}(x) = \frac{1}{[m]_{\pm}} \Big( t^{m-1_\pm}_{j}(x)\, t_j^{1}(x q^{2m-2}) + \eta_{\pm}^{m-1}(x)\, t^{m-1_\pm}_{j}(x \,|\, x q^{2m-2}) \Big),
\end{equation} 
where
\begin{equation} \label{etax}
\eta_{\pm}^m(x) = Q_+ \frac{q^m (q^2-q^{2 m} x^2/b) [m]}{q^2-q^{4 m} x^2} \,\Big|_{q\to \pm q^{\pm1}} \,,
\end{equation}
and
\begin{equation} \label{TMex}
t^{m_\pm}_{j}(x\,|\, x q^{2m}) = \dl \bB^{m_\pm}_j(x) \, \tau_j(x q^{2m}) \dr^{j+m-1}_{j} \;.
\end{equation}
\end{thrm}
We will give a proof of this Theorem a little bit further. For now we want to state some properties of the extended transfer matrix $t^{m_\pm}_{j}(x\,|\, x q^{2m})$. Let us introduce the following short-hand notation,
\begin{equation} \label{TMex2}
t^{\cT_m}_{j}(x \,|\, x c_{m+1} \,|\,\ldots\,|\,xc_{m+k}) = \dl \bB^{\cT_m}_j(x) \,\tau_j(x c_{m+1} )\, \tau_j(x c_{m+2} )\, \cdots \tau_j(x c_{m+k} ) \dr^{j+m-1}_j ,
\end{equation}
and
\begin{equation} \label{TM1ex}
t^{{1}}_{j}(x \,|\, x c_{2} \,|\,\ldots\,|\,xc_{k}) = \dl \tau_j(x c_1)\, \tau_j(x c_2)\cdots \tau_j(x c_{k} ) \dr_j \;.
\end{equation}
The extended operators \eqref{TMex2} respect the following formal recurrence relation, which follows by a direct application of Theorem \ref{th52}:
\begin{crl}
The extended transfer matrices satisfy the following formal recurrence relation 
\begin{align} \label{TMrec_ex}
t^{m_\pm}_{j}(x \,|\, x c_{m+1} \,|\,\ldots\,|\,x c_{m+k}) &= \frac{1}{[m]_{\pm}} \Big( t^{m_\pm}_{j}(x)\, t_j^{{1}}(x c_m \,|\, x c_{m+1} \,|\,\ldots\,|\,x c_{m+k}) \el 
& \hspace{2cm} + \eta_{\pm}^{m-1}(x)\, t^{m_\pm}_{j}(x \,|\, x c_{m}\,|\, x c_{m+1} \,|\,\ldots\,|\,xc_{m+k}) \Big) \,,
\end{align} 
where $c_i = q^{\pm(2i-2)}$.
\end{crl}
We will now state some properties that we will use in the proof of Theorem \ref{th52} which follows afterwards.
\begin{prop} \label{prop54} 
Let $\vec{c}=(1,q^{\pm2},q^{\pm4},\ldots,q^{\pm(2m-2)})$ denote the quantum contents of a totally symmetric (for $+$) or antisymmetric (for $-$) $m$-tableau. Then the following identities hold $(i<m)$:
\begin{align} 
X_{j+i-1}(x)\, \Phi_{m,j}(\vec{c}) = \Phi_{m,j}(\vec{c})\, X_{j+i-1}(x) &= \mathcal{X}_{\mp}(x)\, \Phi_{m,j}(\vec{c}) \,, \label{prop54:1} \\
X_{j+i-1}(x)\, \Psi_{m,j}(\vec{c}) = \Psi_{m,j}(\vec{c}) \,X_{j+i-1}(x) &= \mathcal{X}_{\mp}(x)\, \Psi_{m,j}(\vec{c}) \,, \label{prop54:2}
\end{align}
where
\begin{equation} \label{calX}
\mathcal{X}_\mp(x)= \pm\frac{q^{\pm1}-xq^{\mp1}}{1-x} \,.
\end{equation}
\end{prop}
\begin{proof}
Note that
\begin{equation} \label{prop51:3}
 X_j(x) X_j(q^{\mp2}) = X_j(q^{\mp2}) X_j(x) = \mathcal{X_\mp}(x) X_j(q^{\mp2}) \,.  
\end{equation}
The operator $\Phi_{m,j}(\vec{c})$ can always be written in a such form that the element $X_{j+i-1}(c_{m-i}/c_{m-i+1})$ for any $i=1\ldots m-1$ would be standing at the very left side of $\Phi_{m,j}(\vec{c})$, or equivalently the element $X_{j+i-1}(c_{i}/c_{i+1})$ would be standing at the very right side. The same property in the reversed order holds for $\Psi_{m,j}(\vec{c})$. Now $c_{i}=q^{\pm(2i-2)}$. Hence $c_{i}/c_{i+1}=c_{m-i}/c_{m-i+1}=q^{\mp2}$. This, by \eqref{prop51:3}, gives \eqref{prop54:1} and \eqref{prop54:2} as required.
\end{proof}
\begin{prop} \label{prop55} 
Let $c_i=q^{\pm(2i-2)}$. Then the following identity holds $(\forall Y_k \in \chH_k, \forall x)$
\begin{equation}
\dl X_k(c_i/c_j)\, Y_k\, X_k(x^2 c_i c_j) \dr_{k+1} = \dl Y_k \dr_k + \gamma_\pm^{j,i}(x)\, Y_k \,,
\end{equation}
where
\begin{equation}
 \gamma_\pm^{j,i}(x) =  \frac{Q_+ \lambda\, q^{2 j} (q^4-q^{4 i} x^2/b)}{(q^{2 j}-q^{2 i})(q^4-q^{2 (i+j)} x^2)}\, \Big|_{q\to q^{\pm1}} \;.
\end{equation}

\end{prop}
\begin{proof} 
This property follows straightforwardly by \eqref{traceXYX} with $x=c_i/c_j$ and $y=x^2 c_i c_j$
\end{proof} 
We are now ready to give a proof of Theorem \ref{th52}.
\begin{proof} 
First we factorize the fused boundary operator in the following way,
\begin{align}
\bB^{m_\pm}_j(x) &= 1/[m]_{\pm}!\, X_j(c_{m-1}/c_m)\cdots X_{j+m-3}(c_{2}/c_m)  \el
& \hspace{1.3cm} \times X_{j+m-2}(c_1/c_m) \, \Phi_{m-1,j}(\vec{c}_\circ)\, B^{m-1_\pm}_j(x)\, X_{j+m-2}(x^2 c_1 c_m) \el
& \hspace{4cm} \times X_{j+m-3}(x^2 c_2 c_m) \cdots X_{j}(x^2 c_{m-1} c_m) \, \tau_j(x c_m)  \,.
\end{align}
Here we have assumed $\vec{c}=(1,q^{\pm2},q^{\pm4},\ldots,q^{\pm(2m-2)})$. Then, using Proposition \ref{prop55}, we have
\begin{align} \label{t_proof_1}
& [m]_{\pm}!\; t^{m_\pm}_{j}(x) = \dl \dl X_j(c_{m-1}/c_m)\cdots \dl X_{j+m-3}(c_{2}/c_m) \el
& \hspace{3cm} \times \dl X_{j+m-2}(c_1/c_m) \Phi_{m-1,j}(\vec{c}_\circ) B^{m-1_\pm}_j(x) X_{j+m-2}(x^2  c_1 c_m) \dr_{j+m-1} \el
& \hspace{5cm} \times X_{j+m-3}(x^2 c_2 c_m) \dr_{j+m-2} \cdots X_{j}(x^2 c_{m-1} c_m) \dr_{j+1} \tau_j(x c_m) \dr_{j}  \el
& = \dl \dl X_j(c_{m-1}/c_m)\cdots \dl X_{j+m-3}(c_{2}/c_m) \dl \Phi_{m-1,j}(\vec{c}_\circ) B^{m-1_\pm}_j(x) \dr_{j+m-2} \el
& \hspace{4cm} \times X_{j+m-3}(x^2 c_2 c_m) \dr_{j+m-2} \cdots X_{j}(x^2 c_{m-1} c_m) \dr_{j+1} \tau_j(x c_m) \dr_{j}  \el
& \quad + \gamma_\pm^{m,1}(x) \dl \dl X_j(c_{m-1}/c_m)\cdots \dl X_{j+m-3}(c_{2}/c_m) \, \Phi_{m-1,j}(\vec{c}_\circ) B^{m-1_\pm}_j(x) \el
& \hspace{4cm} \times X_{j+m-3}(x^2 c_2 c_m) \dr_{j+m-2} \cdots X_{j}(x^2 c_{m-1} c_m) \dr_{j+1} \tau_j(x c_m) \dr_{j}  \;.
\end{align}
The second term in the last equality evaluates to
\begin{align}
\gamma_\pm^{m,1}(x) \, \xi_\pm^{m-2} \, \zeta_\pm^{m,m-2}(x)\, \dl \Phi_{m-1,j}(\vec{c}_\circ)\, B^{m-1_\pm}_j(x) \, \tau_j(x c_m) \dr^{j+m-2}_{j}  \;.
\end{align}
where, by Proposition \ref{prop54},
\begin{align}
\xi_\pm^n &= \mathcal{X}_\pm(c_{m-1}/c_m)\mathcal{X}_\pm(c_{m-2}/c_m) \cdots \mathcal{X}_\pm(c_{m-n}/c_m) = [n\!+\!1]_{\pm} \,, \\
\zeta_{\pm}^{m,n}(x) &= \mathcal{X}_\pm(x^2 c_{m-n} c_m) \cdots \mathcal{X}_\pm(x^2 c_{m-2}c_m)\mathcal{X}_\pm(x^2 c_{m-1}c_m) = \frac{q^{n+6}-q^{4 m-n} x^2}{q^6-q^{4 m} x^2} \;\big|_{q\to\pm q^{\pm1}} \;.
\end{align}
The first term in the last equality of \eqref{t_proof_1}, by Proposition \ref{prop55}, gives
\begin{align} \label{t_proof_2}
& \dl X_j(c_{m-1}/c_m)\cdots \dl X_{j+m-4}(c_{3}/c_m) \dl \Phi_{m-1,j}(\vec{c}_\circ) B^{m-1_\pm}_j(x) \dr^{j+m-2}_{j+m-3}\el
& \hspace{4cm} \times X_{j+m-4}(x^2 c_3 c_m) \dr_{j+m-3} \cdots X_{j}(x^2 c_{m-1} c_m) \dr_{j+1} \tau_j(x c_m) \dr_{j}  \el
& \quad + \gamma_\pm^{m,2}(x) \,\dl X_j(c_{m-1}/c_m)\cdots \dl X_{j+m-4}(c_{3}/c_m) \dl \Phi_{m-1,j}(\vec{c}_\circ) B^{m-1_\pm}_j(x) \dr_{j+m-2}\el
& \hspace{4cm} \times X_{j+m-4}(x^2 c_3 c_m) \dr_{j+m-3} \cdots X_{j}(x^2 c_{m-1} c_m) \dr_{j+1} \tau_j(x c_m) \dr_{j}  \;.
\end{align}
The second term in the last equality evaluates to
\begin{align}
\gamma_\pm^{m,2}(x) \, \xi_\pm^{m-3} \, \zeta_\pm^{m,m-3}(x)\, \dl \Phi_{m-1,j}(\vec{c}_\circ) B^{\cT_{m-1}}_j(x) \, \tau_j(x c_m) \dr^{j+m-2}_{j}  \;.
\end{align}
In the same way as in the previous step the first term in the last equality of \eqref{t_proof_2} gives
\begin{align} \label{t_proof_3}
& \dl X_j(c_{m-1}/c_m)\cdots \dl X_{j+m-5}(c_{4}/c_m) \dl \Phi_{m-1,j}(\vec{c}_\circ) B^{m-1_\pm}_j(x) \dr^{j+m-2}_{j+m-4}\el
& \hspace{1.5cm} \times X_{j+m-5}(x^2 c_4 c_m) \dr_{j+m-4} \cdots X_{j}(x^2 c_{m-1} c_m) \dr_{j+1} \tau_j(x c_m) \dr_{j}  \el
& \quad + \gamma_\pm^{m,3}(x) \, \xi_\pm^{m-4} \, \zeta_\pm^{m,m-4}(x)\, \dl \Phi_{m-1,j}(\vec{c}_\circ) B^{\cT_{m-1}}_j(x) \, \tau_j(x c_m) \dr^{j+m-2}_{j}  \; .
\end{align}
Repeating these steps for total $m-1$ times we get
\begin{align}
[m]_{\pm}!\;t^{m_\pm}_{j}(x) &= \dl \dl \Phi_{m-1,j}(\vec{c}_\circ) \, B^{m-1_\pm}_j(x) \dr^{j+m-2}_{j} \,\tau_j(x c_m) \dr_{j}  \el
 & \quad + \Bigg( \sum_{k=2\ldots m} \gamma_\pm^{m,k-1}(x) \, \xi_\pm^{m-k} \, \zeta_\pm^{m,m-k}(x) \Bigg) \dl \Phi_{m-1,j}(\vec{c}_\circ) \, B^{m-1_\pm}_j(x) \, \tau_j(x c_m) \dr^{j+m-2}_{j}  \; .
\end{align}
The sum above can be evaluated explicitly giving \eqref{etax}. Note that $\xi_\pm^{0}=\zeta_\pm^{m,0}(x)=1$. In such a way we obtain \eqref{TMrec} as required. 
\end{proof}
%


The transfer matrix for the primitive tableau of size $m=1$ has been found in \cite{Isaev10},
\begin{align}
t^{1}_j(x) &= \dl \tau_j(x)\dr_j = t_{j-1}^{{1}}(x) + \lambda\, Q_+\frac{1-x^2/b}{(1-x)^2} \,\tau_{j-1}(x) . \label{Isaev218}
\end{align}
This expression is a recurrence equation which can be easily solved, giving ($j\geq2$)
\begin{equation}
t^{1}_j(x) = t^{{1}}_1(x) + \lambda\, Q_+\frac{1-x^2/b}{(1-x)^2} \sum_{k=1\ldots j-1} \tau_k(x) .
\end{equation}
Here $\tau_k(x)$ are the rational functions \eqref{RE_sol}. Transfer matrices for small $m>1$ are given in the example bellow. 

\begin{example} \label{ex51} 
We give the explicit form of the transfer matrices $t^{m_\pm}_{j}(x)$ for some simple tableaux of the sequence $2_\pm \subset 3_\pm \subset 4_\pm \subset \ldots \subset m_\pm$ in step-by-step derivation. For $m=2$ we have
\begin{align} \label{tm2}
t^{2_\pm}_{j}(x) &= 1/[2]_\pm\,\dl X_j(c_1/c_2)\tau_j(x)X_j(x^2 c_1 c_2)\tau_j(x c_2)\dr^{j+1}_j \el 
&= 1/[2]_\pm \Big( \dl \dl \tau_j(x) \dr_{j} \tau_j(x c_2)\dr_j + \gamma^{2,1}_{\pm}(x)\, \dl \tau_j(x)\, \tau_j(x q^{\pm2}) \dr_j \Big) \el
&= 1/[2]_\pm \Big( t^{{1}}_j(x)\, t^{{1}}_j(x q^{\pm2}) + \eta^{1}_{\pm}(x)\, t^1_{j}(x\,|\,xq^{\pm2}) \Big) \,. \hspace{5.48cm}
\end{align}
For $m=3$ we have
\begin{align} \label{tm3}
t^{3_\pm}_{j}(x) &= 1/[3]_\pm!\,\dl X_j(c_2/c_3)X_{j+1}(c_1/c_3) X_{j}(c_1/c_2) B^{2_\pm}_j(x) X_{j+1}(x^2 c_1 c_3)X_j(x^2 c_2 c_3)\tau_j(x c_3)\dr^{j+2}_j \el
& = 1/[3]_\pm! \Big( \dl X_j(c_2/c_3) \dl X_{j}(c_1/c_2) B^{2_\pm}_j(x) \dr_{j+1} X_j(x^2 c_2 c_3)\tau_j(x c_3)\dr^{j+1}_j \el
& \hspace{2cm} + \gamma^{3,1}_{\pm}(x)\, \xi^{1}_{\pm}\, \zeta^{3,1}_{\pm}(x) \dl X_{j}(c_1/c_2) B^{\cT_{2}}_j(x) \tau_j(x c_3)\dr^{j+1}_j \Big) \el
& = 1/[3]_\pm! \Big( \dl \dl X_{j}(c_1/c_2) B^{2_\pm}_j(x) \dr^{j+1}_j \tau_j(x c_3)\dr_j \el
& \hspace{2cm} + \big( \gamma^{3,2}_{\pm}(x) + \gamma^{3,1}_{\pm}(x)\, \xi^{1}_{\pm}\, \zeta^{3,1}_{\pm}(x) \big) \dl X_{j}(c_1/c_2) B^{2_\pm}_j(x) \tau_j(x c_3)\dr^{j+1}_j \Big) \el
& = 1/[3]_\pm \Big( t^{{2}}_{j,\pm}(x)\,t^{{1}}_j(x q^{\pm4}) + \eta^2_{\pm}(x)\, t^2_{j,\pm}(x\,|\,xq^{\pm4}) \Big) \,.
\end{align}
The last operator by \eqref{tm2} evaluates to
\begin{align} \label{tm21}
t^{2_\pm}_{j}(x\,|\,xq^{\pm4}) = 1/[2]_\pm \Big( t^{1}_j(x)\, t^{1}_j(x q^{\pm2} \,|\, x q^{\pm4}) + \eta^{1}_{\pm}(x)\, t^1_{j}(x\,|\,xq^{\pm2}\,|\, x q^{\pm4}) \Big).
\end{align}
For $m=4$ we have (here we have omitted the explicit derivation)
\begin{align} \label{tm4}
t^{4_\pm}_{j}(x) &= 1/[4]_\pm \Big( t^{3_\pm}_{j}(x)\,t^{1}_j(x q^{\pm6}) + \eta^3_{\pm}(x)\, t^{3_\pm}_{j}(x\,|\,xq^{\pm6}) \Big) .\quad\;
\end{align}
The last operator by \eqref{tm3} and \eqref{tm21} evaluates to
\begin{align}
t^{3_\pm}_{j}(x\,|\,xq^{\pm6}) &= 1/[3]_\pm \Big( t^{2_\pm}_{j}(x)\,t^{1}_j(x q^{\pm4}\,|\,x q^{\pm6}) + \eta^2_{\pm}(x)\, t^{2_\pm}_{j}(x\,|\,xq^{\pm4}\,|\,x q^{\pm6}) \Big) \el 
& = 1/[3]_\pm! \Big( [2]_\pm\, t^{2_\pm}_{j}(x)\,t^{1}_j(x q^{\pm4}\,|\,x q^{\pm6}) + \eta^2_{\pm}(x) \big( t^{1}_j(x)\, t^{1}_j(x q^{\pm2} \,|\, x q^{\pm4}\,|\, x q^{\pm6}) \el 
& \hspace{7cm} + \eta^{1}_{\pm}(x)\, t^1_{j}(x\,|\,xq^{\pm2}\,|\, x q^{\pm4}\,|\, x q^{\pm6}) \big) \Big) . 
\end{align}
with $t^{2_\pm}_{j}(x)$ given by \eqref{tm2}.
\end{example}

\smallskip


Let $\cT_{m+1}=\{1,2,\ldots,m;m+1\}$. We will call such a tableaux the $(m,1)$--hook symmetric tableau. Similarly, we will call $\cT_{m+1}=\{1,m+1;2,3,\ldots,m\}$ the $(m,1)$--hook antisymmetric tableau. We denote the corresponding transfer matrices by $t^{m;1_\pm}_j(x)$, where the subscript $+$ (resp.~$-$) denotes totally symmetric (resp.~antisymmetric) case. 


\begin{thrm} \label{th53} 
Transfer matrices for $m\geq2$ satisfy the the following identity
\begin{equation} \label{TMfunc}
t_{j}^{m_\pm}(xq^{\pm1})\,t^1_j(xq^{\mp1}) = \phi^{m}_\pm(x)\, t^{m+1_\pm}_{j}(xq^{\mp1}) + \psi^{m}_\pm(x)\, t^{m;1_\pm}_{j}(xq^{\pm1}) ,
\end{equation} 
where
\begin{equation} \label{pfipsi}
\phi^{m}_\pm(x) =  \frac{1-q^{2m} x^2/b}{q^m(1-x^2/b)} \,\Big|_{q\to\pm q^{\pm1}} , \qquad \psi^m_\pm(x) = -\frac{(q^2-x^2/b)[m]}{q(1-x^2/b)} \,\Big|_{q\to\pm q^{\pm1}}.
\end{equation}
\end{thrm}
In terms of Young tableaux the identity above can be represented as ($k=2m-1$)
\begin{equation*}
\ytableausetup{boxsize=1.2em,aligntableaux=center}
\ytableaushort{{{\text{\scriptsize $q$}}}{{\text{\scriptsize $q^3$}}}{{\text{\scriptsize $...$}}}{{\text{\scriptsize $q^{k}$}}}}
\;\times\;
\ytableaushort{{{\text{\scriptsize $q^{\text{-}1}$}}}}
\;=\;
\ytableaushort{{{\text{\scriptsize $q^{\text{-}1}$}}}{{\text{\scriptsize $q$}}}{{\text{\scriptsize $q^3$}}}{{\text{\scriptsize $...$}}}{{\text{\scriptsize $q^k$}}}}
\;+\;
\ytableaushort{{{\text{\scriptsize $q$}}}{{\text{\scriptsize $q^3$}}}{{\text{\scriptsize $...$}}}{{\text{\scriptsize $q^k$}}},{{\text{\scriptsize $q^{\text{-}1}$}}}}
\end{equation*}
for the `$+$' case, and 
\begin{equation*}
\ytableausetup{boxsize=1.2em,aligntableaux=center}
\ytableaushort{{{\text{\scriptsize $q^{\text{-}1}$}}},{{\text{\scriptsize $q^{\text{-}3}$}}},{{\text{\scriptsize $...$}}},{{\text{\scriptsize $q^{\text{-}k}$}}}}
\;\times\;
\ytableaushort{{{\text{\scriptsize $q$}}}}
\;=\;
\ytableaushort{{{\text{\scriptsize $q$}}},{{\text{\scriptsize $q^{\text{-}1}$}}},{{\text{\scriptsize $q^{\text{-}3}$}}},{{\text{\scriptsize $...$}}},{{\text{\scriptsize $q^{\text{-}k}$}}}}
\;+\;
\ytableaushort{{{\text{\scriptsize $q^{\text{-}1}$}}}{{\text{\scriptsize $q$}}},{{\text{\scriptsize $q^{\text{-}3}$}}},{{\text{\scriptsize $...$}}},{{\text{\scriptsize $q^{\text{-}k}$}}}}
\end{equation*}
for the `$-$' case.
\begin{proof} 
First, by Proposition \ref{prop52}, we have
\begin{align}
t^{m_\pm}_{j}(xq^{\pm1})\, t^1_j(xq^{\mp1}) &= \frac{1}{[m]_{\pm}!} \dl \Phi_{m,j}(\vec{c})\,B^{m_\pm}_j(xq^{\pm1}) \dr^{j+m-1}_j \dl \tau_j (xq^{\mp1})\dr_j \el
&= \frac{1}{[m]_{\pm}!} \dl \Phi_{m,j}(\vec{c})\,B^{m_\pm}_j(xq^{\pm1}) \cev{X}_{m,j}(x^2\vec{c}) \, \tau_j (xq^{\mp1}) \vec{X}_{m,j}(b/(x^2 \vec{c})) \dr^{j+m}_j ,
\end{align}
where we have assumed $\vec{c}=(1,q^{\pm2},q^{\pm4},\ldots,q^{\pm(2m-2)})$.
Next, we insert a partition of unity into the front of the trace above,
\begin{equation}
1 = A^{m+1_+}_j + A^{m;1_+}_j + \ldots + A^{m;1_-}_j + A^{m+1_-}_j.
\end{equation}
%
%
Due to orthogonality with $A^{m_\pm}_j$ only the first two ($+$ case) or the last two ($-$ case) elements contribute. The insertion of the element $A^{m+1_\pm}_j$, by \eqref{A_IMO_sa} and Proposition \ref{prop54}, gives a factor 
\begin{equation}
\phi^{m}_\pm(x) = \mathcal{X}_\mp(b/x^2)\mathcal{X}_\mp(b/(x^2q^{\pm2}))\cdots\mathcal{X}_\mp(b/(x^2q^{\pm(2m-2)})) = \frac{1-q^{2m} x^2/b}{q^m(1-x^2/b)} \,\Big|_{q\to\pm q^{\pm1}} \,,
\end{equation}
times an element which, upon insertion of $g_\pm(\vec{c})\,\vec{X}_{m,j}(q^{\mp2}/\cev{c})$, evaluates to the required transfer matrix 
\begin{align}
& \dl A^{m+1_\pm}_j\, B^{m_\pm}_j(xq^{\pm1}) \cev{X}_{m,j}(x^2\vec{c}) \, \tau_j (xq^{\mp1}) \,g_\pm(\vec{c})\,\vec{X}_{m,j}(q^{\mp2}/\cev{c}) \dr^{j+m}_j \el
& \qquad = \dl A^{m+1_\pm}_j\,g_\pm(\vec{c})\,\cev{X}_{m,j}(q^{\mp2}/\cev{c}) \, \tau_j (xq^{\mp1})\,\vec{X}_{m,j}(x^2\vec{c})\, B^{m_\pm}_j(xq^{\pm1})  \, \cev{X}_{m,j}(q^{\pm2}\vec{c})) \dr^{j+m}_j \el
& \qquad = \dl A^{m+1_\pm}_j B^{m+1_\pm}_j(xq^{\mp1}) \dr^{j+m}_j = t^{m+1_\pm}_{j}(xq^{\mp1}) \,.
\end{align}
Here we have employed the fused RE, and $g_\pm(\vec{c})$ is a polynomial in $q$ such that the following identity (by Proposition \ref{prop54}) holds,
\begin{equation}
g_\pm(\vec{c})\,A^{m+1_\pm}_j\,\cev{X}_{m,j}(q^{\mp2}/\cev{c}) = g_\pm(\vec{c})\,\vec{X}_{m,j}(q^{\mp2}/\cev{c}) \,A^{m+1_\pm}_j= A^{m+1_\pm}_j \,.
\end{equation}
%

%
%
Next, we write $A^{m;1_\pm}_j = F(\lambda)^2 \, \Psi_{m+1,j}(\vec{c})\, X_{w_{m+1,j}}^{-2} \Phi_{m+1,j}(\vec{c})$, where $\vec{c}=(1,q^{\pm2},\ldots,q^{\pm(2m-2)},q^{\mp2})$ and $\lambda$ is the shape of the tableau $m+1_\pm$. Then, by inserting this element, we obtain the second required transfer matrix,
\begin{align}
& \dl A^{m;1_\pm}_j B^{m;1_\pm}_j(xq^{\pm1}) \vX_{m,j}(b/(x^2 \vec{c})) \dr^{j+m}_j \el
& \quad = F(\lambda)^2 \dl \Psi_{m+1,j}(\vec{c})\, X_{w_{m+1,j}}^{-2} \Phi_{m+1,j}(\vec{c})\, B^{m;1_\pm}_j(xq^{\pm1})  \vX_{m,j}(b/(x^2 \vec{c})) \dr^{j+m}_j \el
& \quad = F(\lambda)^2 \dl X_{w_{m+1,j}}^{-2} B^{m;1_\pm}_j(xq^{\pm1})^R \, \Psi_{m+1,j}(\vec{c})\,  \vX_{m,j}(b/(x^2 \vec{c})) \, \Psi_{m+1,j}(\vec{c}) \dr^{j+m}_j \el
& \quad = \psi^{m}_\pm(x) \, \dl A^{m;1_\pm}_j B^{m;1_\pm}_j(xq^{\pm1}) \dr^{j+m}_j = \psi^{m}_\pm(x) \, t^{m;1_\pm}_{j}(xq^{\pm1}) \,,
\end{align}
where, by similar considerations as in Proposition \ref{prop54},
\begin{equation}
\psi^{m}_\pm(x) = \mathcal{X}_\mp(q^{\mp(2m-2)}) \cdots \mathcal{X}_\mp(q^{\mp4})\mathcal{X}_\mp(q^{\mp2})\mathcal{X}_\pm(b/x^2) = -\frac{(q^2-x^2/b)[m]}{q(1-x^2/b)}\,\Big|_{q\to\pm q^{\pm1}} \,.
\end{equation}
Finally, by combining all the terms together we obtain \eqref{TMfunc} as required.
\end{proof}

\begin{remark}
The Hamiltonian of an open Hecke chain can be obtained (up to an overall normalization factor and a constant) by differentiating the corresponding transfer matrix with respect to the spectral parameter $x$ and evaluating at $x=1$ (see Section 3 in \cite{Isaev05} for details).
\end{remark}

\subsection{Some low-rank applications}

In this section we will briefly consider the free boundary model. Here we will specialize the methods and expressions derived above to this particular case, and state the explicit form of the extended primitive transfer matrices for some low-rank cases.

The free boundary model is defined by
\begin{equation}
\tau_j(x)= X_{j-1}(x) \cdots X_{1}(x)\,\tau_1(x)\,X_{1}(x) \cdots X_{j-1}(x), \qquad \tau_1(x)=1.
\end{equation}
This model is particularly elegant due to the simplicity of the $m=j=1$ transfer matrices, namely $t^1_1(x \,|\, x q^{\pm2} \,|\, \ldots ) =  t^1_1(x) = Q_0$. 
In such a way $t^{m_\pm}_{1}(x)$ and $t^{m,1_\pm}_{1}(x)$ are rather simple polynomials in $Q_0$, 
\begin{equation} \label{TM_free}
t^{m_\pm}_{1}(x) = \frac{Q_0}{[m]_\pm!} \prod_{k=1\ldots m-1} \big( Q_0 + \eta_\pm^k(x) \big) ,
\end{equation} 
and
\begin{equation} \label{TM_hook_free}
t^{m,1_\pm}_{1}(x) = 1/\psi^m_\pm(x) \Big( t^{m_\pm}_1(xq^{\pm1})\, Q_0 + \phi^m_\pm(x)\, t^{m+1_\pm}_1(xq^{\mp1}) \Big) .
\end{equation}
The explicit form of these elements for some small $m$ are given in the example below. 
%
%
\begin{example} 
Set $j=1$. Then
\begin{align}
t^1_{1}(x) \;\; &= Q_0 \,, \\
t^{2_\pm}_{1}(x) &= \frac{Q^2_0}{[2]!}\, \frac{1-x^2}{1-b} \,\frac{1-b q^2}{1-q^2 x^2} \;\Big|_{q\to\pm q^{\pm1}}\,, \\
t^{3_\pm}_{1}(x) &= \frac{Q^3_0}{[3]!} \, \frac{1-x^2}{(1-b)^2} \, \frac{(1-b q^2) (1-b q^4)}{1-q^6 x^2} \;\Big|_{q\to\pm q^{\pm1}}\,, \\
t^{4_\pm}_{1}(x) &= \frac{Q^4_0}{[4]!} \, \frac{1-x^2}{(1-b)^3} \, \frac{(1-b q^2) (1-b q^4)(1-b q^6)}{(1-q^6 x^2)(1-q^{10}x^2)} \;\Big|_{q\to\pm q^{\pm1}}\,, 
\end{align}
and
\begin{align}
t^{2,1_\pm}_1(x) &= \frac{Q^3_0}{[3]!} \,\frac{b-(1+b^2) q^2+b q^4}{q^2(1-b)^2} \;\Big|_{q\to\pm q^{\pm1}}\,, \label{tm21free}\\
t^{3,1_\pm}_1(x) &= \frac{Q^4_0}{[4]!} \,\frac{(b-q^2)(1-b q^2)(1-b q^4)(1-x^2)(1-q^4x^2)}{q^2(1-b)^3 (1-q^2 x^2)(1-q^6 x^2)} \;\Big|_{q\to\pm q^{\pm1}}\,. \label{tm31free}
\end{align}
\end{example}

The (extended) primitive transfer matrices for $j=2$ are polynomial functions in the spectral parameter $x$, namely
\begin{equation}
	t^1_2(x \,|\, x c_2\,|\,\ldots\,|\, x c_k) = Q_0\, (1+\lambda_2(x,\vec{c})) .
\end{equation}
\begin{example}
Functions $\lambda_2(x,\vec{c})$ for $\vec{c}_{(1)}=1$ and $\vec{c}_{(2)}=(1,c_2)$ are
\begin{align}
\lambda_2(x,\vec{c}_{(1)}) &= \frac{\lambda^2 (x^2-b)}{(1-b)(1-x)^2} \,,\\
\lambda_2(x,\vec{c}_{(2)}) &= \frac{\lambda ^2 (x^3 (x ((1-b) q^2+q^4+1)c_2^2/q^2 - 2 c_2^2)+2 c_2 x (b-x^2)-b ([3]-2 x)+1)}{(1-b) (1-x)^2(1-c_2 x)^2} \,. 
\end{align}
Here the first case gives the primitive transfer matrix $t^1_2(x)$.
\end{example} 

For $j\geq3$ primitive operator $t^1_j(x)$ has the following rather simple form (see also \cite{Isaev10})
\begin{equation}
t^1_j(x) = Q_0 + \lambda\, Q_+ \frac{1-x^2/b}{(1-x)^2}\, J_j(x) \,,
\end{equation}
where 
\begin{equation}
J_j(x) = 1 + \sum_{k=1\ldots(j-2)} \big( X_k(x) \cdots X^2_1(x) \cdots X_k(x) \big) = f_j(x) +\frac{1+x}{x} \sum_{i=1\ldots(2j-5)} \left(\frac{\lambda\,x}{1-x}\right)^i J_i \,.
\end{equation}
Here $\{J_i\}$ is the set of commuting elements in $\mathcal{H}_{j-1}$. These elements can be obtained by considering the power series expansion of $J_j(x)$ in the neighbourhood of $x=0$. For example, for $j=3$ it is $J_1=\sigma_1$ only, for $j=4$ it is $J_1=\sigma_1 + \sigma_2$, $J_2=\sigma_1\sigma_2 + \sigma_1\sigma_2$, $J_3=\sigma_1 + \sigma_2+\sigma_1\sigma_2\sigma_1$.
Functions $f_j(x)$ are scalar polynomials in $x$,
\begin{equation}
f_j(x) = 1 + \sum_{i=1\ldots(j-2)} {j\!-\!1 \choose i\!+\!1} \left(\frac{\lambda\,x}{1-x}\right)^{2i} \,,
\end{equation}
where ${j-1 \choose i+1}$ are binomial coefficients. In such a way the primitive transfer matrices have the generic form
\begin{equation}
t^1_j(x) = Q_0 \sum_{i=0\ldots (2j-5)} \lambda_{i,j}(x) \, J_i \,,
\end{equation}
where $J_0=1$ and $\lambda_{i,j}(x)$ are scalar polynomial functions in $x$. The extended primitive operators for $j\geq3$ have the same generic form,
\begin{equation}
t^1_j(x \,|\, x c_2\,|\,\ldots\,|\, x c_k) = Q_0 \sum_{i=0\ldots(2j-5)} \lambda_{i,j}(x,\vec{c})\, J_i \,,
\end{equation}
where $\lambda_{i,j}(x,\vec{c})$ are again scalar polynomial functions in $x$. It is a straightforward calculation to find these functions for a given $\vec{c}$, however the expressions are very bulky. 

\newpage

\section{Conclusions}

In this work we have demonstrated that the fusion procedure can be extended to arbitrary representations of the Hecke algebra such that fused baxterized elements satisfy fused Yang-Baxter and Reflection equations. The general rational solution (known as the baxterized boundary operator) of the baxterized affine Hecke algebra for the fundamental representation was found in \cite{Isaev05} and further generalized to totally (anti-) symmetric representations in \cite{Isaev10}. Operators of this type have also been considered in \cite{IsMoOg11B}. Here we have constructed generalized baxterized boundary operators for arbitrary representations and proven them to be solutions of the generalized Sklyanin reflection equation. These operators could further be linked to known and new quantum integrable systems defined on a half-line.  

The main objective of this work was to construct transfer matrices of the Sklyanin type for arbitrary representations of the affine Hecke algebra, and derive fusion type functional relations for these operators for selected families of representations. The results obtained generalize the constructions and relations found in \cite{Isaev10} and could further be extended to a wider class of representations. These functional relations will serve as a good starting point for an investigation of the spectrum of Hamiltonian systems associated with the affine Hecke algebra. The Hamiltonians (up to an overall normalization factor and a constant) can be obtained by differentiating the corresponding transfer matrices with respect to the spectral parameter $x$ and evaluating at $x = 1$.
 
Functional relations deserve further investigation and generalizations. An obvious extensions of this work is to find a class of Young diagrams and tableaux for which functional relations form a closed set of equations and can be solved exactly. Another topic for future research is the explicit form of the associated Hamiltonians and the complete set of conserved charges corresponding to the new transfer matrices. A challenging problem is to understand similarities and differences between the structure of functional relations for the affine (boundary) Hecke algebra and its much more studied Lie algebraic (bulk) analogues \cite{KNS11}. 

\bigskip

\noindent {\bf Acknowledgements.} 
The authors are thankful to  N.~J.~MacKay for useful discussions and to A.~I.~Molev for valuable comments. We also thank the UK EPSRC for funding under grants EP/H000054/1 and EP/K031805/1. A.~B. is also grateful to the Hebrew University where this work has been completed.


\appendix

\section{Some low-rank functional relations}

In this appendix we give an explicit derivation of some simple functional relations that serve both as examples of the techniques developed in this paper and also as simple checks of the recurrence relations and fusion rules given by Theorems \ref{th52} and \ref{th53}.  

\begin{example} 
Let us explicitly derive a functional relation for the primitive transfer matrices,
\begin{equation}
t^1_j(x q)\,t^1_j(x/q) = \phi^1_+(x)\, t^{2_+}_{j}(x/q) + \psi^1_+(x)\, t^{2_-}_{j}(x q) \,,
\end{equation}
where (see \eqref{prop54:2} and \eqref{pfipsi})
\begin{equation}
\phi^1_+(x) = \mathcal{X}_+(x) = -\frac{q^{-1}-xq}{1-x} \,, \qquad\quad
\psi^1_+(x) = \mathcal{X}_-(x) = \frac{q-xq^{-1}}{1-x} \,.
\end{equation}
In terms of Young tableaux this functional relation reads as
\begin{equation*}
\ytableausetup{boxsize=1.2em,aligntableaux=center}
\ytableaushort{{{\text{\scriptsize $q$}}}}
\;\times\;
\ytableaushort{{{\text{\scriptsize $q^{\text{-}1}$}}}}
\;=\;
\ytableaushort{{{\text{\scriptsize $q^{\text{-}1}$}}}{{\text{\scriptsize $q$}}}}
\;+\;
\ytableaushort{{{\text{\scriptsize $q$}}},{{\text{\scriptsize $q^{\text{-}1}$}}}}
\end{equation*}
First note that $X_j(q^{-2})-X_j(q^2)=[2]$. Then, by \eqref{CUB},
\begin{align}
& [2]\, t^1_j(xq)\, t^1_j(x/q) = [2]\,\dl \tau_j (x q)\dr_j \dl \tau_j (x/q)\dr_j \el
& \quad = \dl (X_j(q^{-2})-X_j(q^{2}))\,\tau_j (x q) X_j(x^2) \tau_j (x/q) X_j(b/x^2)\dr^{j+1}_j \el
& \quad = \mathcal{X}_-(b/x^2) \dl X_j(q^{-2})\tau_j (x/q) X_j(x^2) \tau_j (xq) \dr^{j+1}_j - \mathcal{X}_+(b/x^2) \dl X_j(q^{2})\tau_j (xq) X_j(x^2) \tau_j (x/q) \dr^{j+1}_j \el
& \quad = [2]\,\big( \mathcal{X}_-(b/x^2)\,t^{2_+}_{j}(x/q) + \mathcal{X}_+(b/x^2)\, t^{2_-}_{j}(x q) \big) ,
\end{align}
where we have used cyclicity of the trace together with RE to obtain $t^{2_+}_{j}(x/q)$.

This relation can be easily checked explicitly for the free boundary model. Setting $j=1$ we obtain the identity
\begin{equation}
Q_0^2 = \frac{Q_0}{\,[2]} \Big( \mathcal{X}_-(b/x^2)\, (Q_0+ \eta^2_+(x/q)) - \mathcal{X}_+(b/x^2)\, (Q_0+\eta^2_-(x q)) \Big) ,
\end{equation}
where, by Theorem \ref{th52},
\begin{equation}
\eta^2_+(x/q) = Q_+\frac{q^2- x^2/b}{q(1-x^2)} \,, \qquad 
\eta^2_-(x q) = -Q_+\frac{1-q^2 x^2/b}{q(1-x^2)} \,.
\end{equation}
\end{example}

\smallskip

\begin{example} 
Let us explicitly derive a functional relation for the primitive hook transfer matrix,
\begin{equation} \label{FR_adjoint}
t^{2_+}_j(x q)\,t^{1}_j(x/q) = \phi^2_+(x)\, t^{3_+}_{j}(x/q) + \psi^2_+(x)\, t^{2,1_+}_{j}(x q) ,
\end{equation}
where (see \eqref{pfipsi})
\begin{equation} \label{phipsi2}
\phi^2_+(x) = \frac{1-q^4x^2/b}{q^2(1-x^2/b)} \,, \qquad
\psi^2_+(x) = -\frac{(q^2 - x^2/b)[2]}{q(1-x^2/b)} \,.
\end{equation}
In terms of Young tableaux this functional relation reads as
\begin{equation*}
\ytableausetup{boxsize=1.2em,aligntableaux=center}
\ytableaushort{{{\text{\scriptsize $q$}}}{{\text{\scriptsize $q^3$}}}}
\;\times\;
\ytableaushort{{{\text{\scriptsize $q^{\text{-}1}$}}}}
\;=\;
\ytableaushort{{{\text{\scriptsize $q^{\text{-}1}$}}}{{\text{\scriptsize $q$}}}{{\text{\scriptsize $q^3$}}}}
\;+\;
\ytableaushort{{{\text{\scriptsize $q$}}}{{\text{\scriptsize $q^3$}}},{{\text{\scriptsize $q^{\text{-}1}$}}}}
\end{equation*}
First, by \eqref{CUB}, we have
\begin{align}
& t^{2_+}_j(xq)\, t^1_j(x/q) = 1/[2] \, \dl X_j(q^{-2}) B^{2_+}_j(xq) \dr^{j+1}_j \dl \tau_j (x/q)\dr_j \el
& \quad = 1/[2]\, \dl X_j(q^{-2})B^{2_+}_j(xq) X_{j+1}(x^2) X_j(x^2q^2) \tau_j (x/q) X_{j}(b/(x^2q^2))X_{j+1}(b/x^2) \dr^{j+2}_j .
\end{align}
Next, we insert a partition of unity into the front of the trace above,
\begin{equation}
1 = A^{(1,2,3)}_j + A^{(1,2;3)}_j + A^{(1,3;2)}_j + A^{(1;2;3)}_j .
\end{equation}
Due to orthogonality condition only the first two elements contribute. The insertion of the idempotent $A^{(1,2,3)}_j=1/[3]!\,\Phi_{3,j}(\vec{c})$ where $\vec{c}=(1,q^2,q^4)$, by \eqref{A_IMO_sa} and \eqref{prop54:1}, gives the required factor \mbox{$\phi^2_+(x) = \mathcal{X}_-(b/(x^2q^2))\mathcal{X}_-(b/x^2)$} times an element which, by application of Proposition \ref{prop54} and fused RE \eqref{fRE}, evaluates to the required transfer matrix,
\begin{align}
& 1/[3]!\, \dl \Phi_{3,j}(\vec{c}) \, B^{2_+}_j(xq) X_{j+1}(x^2) X_j(x^2q^2) \tau_j (x/q) \dr^{j+2}_j \el
& \quad = 1/[3]!\, \dl \Phi_{3,j}(\vec{c}) \, B^{2_+}_j(xq) X_{j+1}(x^2) X_j(x^2q^2) \tau_j (x/q)  X_j(q^{-4}) X_{j+1}(q^{-2}) \dr^{j+2}_j / (\mathcal{X}_-(q^{-2})\mathcal{X}_-(q^{-4})) \el
& \quad = 1/[3]!\, \dl \Phi_{3,j}(\vec{c}) \, \tau_j (x/q) X_{j}(x^2) X_{j+1}(x^2q^2) \, B^{2_+}_j(xq)  \dr^{j+2}_j \el
& \quad = 1/[3]!\, \dl \Phi_{3,j}(\vec{c}) \, B^{3_+}_j(x/q) \dr^{j+2}_j = t^{3_+}_j(x/q) \,.
\end{align}
The insertion of the idempotent $A^{2,1_+}_j= 1/[3]\,X^{-1}_{w_{3,j}} \Phi_{3,j}(\vec{c})$ where $\vec{c}=(1,q^2,q^{-2})$ gives
\begin{align}
& 1/[3]\, \dl X^{-1}_{w_{3,j}} \Phi_{3,j}(\vec{c})\, B^{2_+}_j(xq) X_{j+1}(x^2) X_j(x^2q^2) \tau_j (x/q) X_{j}(b/(x^2q^2)) X_{j+1}(b/x^2) \dr^{j+2}_j \el
& \quad = 1/[3]^2\, \dl \Psi_{3,j}(\vec{c}) X^{-2}_{w_{3,j}} \Phi_{3,j}(\vec{c})\, B^{2,1_+}_j(xq) X_{j}(b/(x^2q^2)) X_{j+1}(b/x^2) \dr^{j+2}_j \el
& \quad = 1/[3]^2\, \dl X^{-2}_{w_{3,j}} \, B^{2,1_+^R}_j(xq)\, \Psi_{3,j}(\vec{c}) X_{j}(b/(x^2q^2)) X_{j+1}(b/x^2) \Psi_{3,j}(\vec{c}) \dr^{j+2}_j \el
& \quad = \psi^2_+(x)\,\dl A^{2,1_+}_j \, B^{2,1_+}_j(xq) \dr^{j+2}_j = \psi^2_+(x)\,t^{2,1_+}_j(xq) \,.
\end{align}
with $\psi^2_+(x)$ given by \eqref{phipsi2}. Here in the first equality we have used the idempotence property and in the second equality we have used Lemma \ref{lemma42}. The last step is a direct calculation. In such a way we have obtained \eqref{FR_adjoint} as required.

The element $t^{2,1_+}_j(xq)$ may be evaluated explicitly. The idempotent $A^{2,1_+}_j$ can be decomposed in the following way,
\begin{equation}
A^{2,1_+}_j = -\frac{1}{[2][3]}X_j(q^{4})X_{j+1}(q^{2})X_j(q^{-2})-\frac{1}{[2]^2}X_{j+1}(q^{2})X_{j}(q^{-2}) \,.
\end{equation}
Then
\begin{equation}
t^{2,1_+}_j(xq) = - \frac{1}{[2][3]} \dl X_j(q^{4})X_{j+1}(q^{2})X_j(q^{-2}) B^{2,1_+}_j(xq) \dr^{j+2}_j -\frac{1}{[2]^2} \dl X_{j+1}(q^{2})X_{j}(q^{-2}) B^{2,1_+}_j(xq) \dr^{j+2}_j .
\end{equation}
The first trace evaluates to
\begin{align}
& \dl X_j(q^{4})\dl X_j(q^{-2}) B^{2_+}_j(xq) \dr_{j+1} X_j(x^2q^2) \tau(x/q) \dr^{j+1}_j \el
& \qquad - Q_+\frac{1-q^2 x^2/b}{q(1-x^2)}
\dl X_j(q^{4}) X_j(q^{-2}) B^{2_+}_j(xq) X_j(x^2q^2) \tau(x/q) \dr^{j+1}_j \el
& = [2]\,t^{2_+}_j(xq) \,t^1(x/q) - \bigg( Q_+ \frac{1-q^6 x^2/b}{q^2(1-q^2 x^2)} + {Q_+}\frac{1-q^2 x^2/b}{1-q^2x^2} \bigg) \,t^{2_+}_j(xq\,|\,x/q) \el
& = [2]\,\Big( t^{2_+}_j(xq) \,t^1(x/q) + \eta^1_-(xq^2) \,t^{2_+}_j(xq\,|\,x/q) \Big) ,
\end{align}
where (see \eqref{etax})
\begin{equation}
\eta^1_-(xq^2)= - Q_+ \frac{1-q^4 x^2/b}{q(1-q^2 x^2)} \,.
\end{equation}
The second trace vanishes,
\begin{align}
& \dl \dl X_j(q^{-2}) B^{2_+}_j(xq) \dr_{j+1} X_j(x^2q^2) \tau(x/q) \dr^{j+1}_j \el
& \qquad - Q_+\frac{1-q^2 x^2/b}{q(1-x^2)}
\dl X_j(q^{-2}) B^{2_+}_j(xq) X_j(x^2q^2) \tau(x/q) \dr^{j+1}_j \el
& = \bigg( Q_+\frac{1-q^2 x^2/b}{1-q^2 x^2} - Q_+\frac{1-q^2 x^2/b}{1-q^2x^2}\,\bigg) [2]\,t^{2_+}_j(xq\,|\,x/q) = 0\,.
\end{align}
In such a way we find
\begin{equation}
t^{2,1_+}_j(xq) = -\frac{1}{[3]}\Big( t^{2_+}_j(xq) \,t^1(x/q) + \eta^1_-(xq^2) \,t^{2_+}_j(xq\,|\,x/q) \Big) \,.
\end{equation}
For the free boundary model with $j=1$ this agrees with \eqref{tm21free} as required.
\end{example}

\newpage
\section{Notation summary}

\noindent For readers' convenience we give a brief summary of the notation used in this paper.

\noindent{\it Section 2}. Operators in the Hecke algebra:
\begin{align}  
\vX_{m,j} &= \sigma_j\sigma_{j+1}\cdots\sigma_{j+m-1}, \qquad \cX_{m,j}= \sigma_{j+m-1}\sigma_{j+m-2}\cdots\sigma_{j} , \tag{\ref{Xvec}}
\\
X_{w_{m,j}} &= \cX_{1,j}\, \cX_{2,j} \cdots \cX_{m,j} = \vX_{m,j} \, \vX_{m-1,j} \cdots  \vX_{1,j} , \tag{\ref{WK2}}
\\
\vX_{m,j}(x \vec{u}) &= X_j(x u_{1})\, X_{j+1}(x u_{2}) \cdots X_{j+m-1}(x u_m) , \tag{\ref{Xvec2}} \\
\cX_{m,j}(x \vec{u}) &= X_{j+m-1}(x u_{1})\, X_{j+m-2}(x u_{2}) \cdots X_{j}(x u_m) , \tag{\ref{Xvec3}}
\\ 
\Psi_{m,j}(\vec{u}) &= \prod_{i=1 \ldots m-1}^{\rightarrow} \! \cX_{i,j}(\vec{u}_{(i)}/u_{i+1}) \,, \qquad \Phi_{m,j}(\vec{u}) = \prod_{i=1\ldots m-1}^{\leftarrow} \!\vX_{i,j}(\cev{u}_{(i)}/u_{i+1}) \,. \tag{\ref{PsiPhi}}
\end{align}
Here $\vec{u}=(u_1,u_2,\ldots,u_m)$, $\cev{u}=(u_m,\ldots,u_2,u_1)$ and $\vec{u}_{(i)}=(u_1,u_2,\ldots,u_i)$, $\cev{u}_{(i)}=(u_i,\ldots,u_2,u_1)$.
 
\noindent {\it Section 4}. Operators of the fused reflection equation:
\begin{align}
B^{\cV_m}_j(x) &= \prod_{i=1\ldots m}^{\rightarrow} \cX_{i-1,j} (x^2 u_{i}\, \vec{u}_{(i-1)}) \, \tau_j(x\,u_{i}) \,, \qquad \tag{\ref{FusedB}} 
\\
Y_{j,++}^{\cV_m,\cV'_n}(x) &= \prod_{i=1\ldots m}^{\leftarrow} \vX_{n,j+i-1}(xu_{m-i+1}\vec{v}\,) \;= \prod_{i=1\ldots n}^{\rightarrow} \cX_{m,j+i-1}(x\vec{u}\,v_i) \,, \tag{\ref{Ypp}}\\
Y_{j,+-}^{\cV_m,\cV'_n}(x) &= \prod_{i=1\ldots m}^{\leftarrow} \vX_{n,j+i-1}(xu_{m-i+1}/\cev{v}) = \prod_{i=1\ldots n}^{\rightarrow} \cX_{m,j+i-1}(x\vec{u}/v_{n-i+1}) \,, \tag{\ref{Ypm}} \\
Y_{j,-+}^{\cV_m,\cV'_n}(x) &= \prod_{i=1\ldots m}^{\leftarrow} \vX_{n,j+i-1}(x\vec{v}/u_i) \hspace{.87cm} = \prod_{i=1\ldots n}^{\rightarrow} \cX_{m,j+i-1}(x v_i/\cev{u}) \,, \tag{\ref{Ymp}} \\
Y_{j,--}^{\cV_m,\cV'_n}(x) &= \prod_{i=1\ldots m}^{\leftarrow} \vX_{n,j+i-1}(x/(u_i\cev{v}\,)) \hspace{.51cm} = \prod_{i=1\ldots n}^{\rightarrow} \cX_{m,j+i-1}(x/(\cev{u}\,v_{n-i+1})) \,. \tag{\ref{Ymm}}
\end{align}
Operators of the fused projected reflection equation:
\begin{align}
\bB_j^{\cT_{m}}(x) &= A_j^{\cT_{m}} B_j^{\cT_{m}}(x), \tag{\ref{FusedBA}}\\
 \bY_{j,-+}^{\cT_m,\cT'_n}(x) &=  A_{j}^{\cT_m} A_{j+m}^{\cT'_n}\, Y_{j,-+}^{\cT_m,\cT'_n}(x) \,, &
 \bY_{j,++}^{\cT_m,\cT'_n}(x) &=  A_{j+m}^{\cT'_n} \, Y_{j,++}^{\cT_m,\cT'_n}(x) \,, \tag{{\ref{bY1}},\ref{bY2}}\\
 \bY_{j,+-}^{\cT_m,\cT'_n}(x) &=  Y_{j,+-}^{\cT_m,\cT'_n}(x) \,, &
 \bY_{j,--}^{\cT'_n,\cT_m}(x) &= A^{\cT_m}_j \, Y_{j,--}^{\cT'_n,\cT_m}(x) \,. \tag{\ref{bY3},\ref{bY4}}
\end{align}
{\it Section 5}. Transfer matrices ($\vec{c}\in\cT_m$):
\begin{align}
t_j^{\cV_{m}}(x) = \dl B^{\cV_{m}}_j(x) \dr^{j+m-1}_{j}, & \qquad\qquad
t_j^{\cT_{m}}(x) = \dl \bB^{\cT_{m}}_j(x) \dr^{j+m-1}_{j},  \tag{\ref{TM},\ref{TMA}}
\\
t^{\cT_m}_{j}(x \,|\, x c_{m+1} \,|\,\ldots\,|\,xc_{m+k}) &= \dl \bB^{\cT_m}_j(x) \,\tau_j(x c_{m+1} )\, \tau_j(x c_{m+2} )\, \cdots \tau_j(x c_{m+k} ) \dr^{j+m-1}_j , \tag{\ref{TMex2}}
\\
t^{{1}}_{j}(x \,|\, x c_{2} \,|\,\ldots\,|\,xc_{k}) &= \dl \tau_j(x c_1)\, \tau_j(x c_2)\cdots \tau_j(x c_{k} ) \dr_j \;. \tag{\ref{TM1ex}}
\end{align}
%

\bibliographystyle{amsplain}

\end{document}